\newif\iflncs
\newif\ifnotes

\newif\ifblind\blindfalse

\notesfalse

\lncsfalse

\iflncs
\RequirePackage{amsmath}
\documentclass{llncs} 
\usepackage{StyleSheets/template_lncs}
\usepackage{StyleSheets/global}
\else
\documentclass[runningheads, 11pt]{article}
\usepackage{times}
\usepackage{physics}
\usepackage{StyleSheets/template}
\usepackage{StyleSheets/global}
\fi

\usepackage{global}
\usepackage{graphicx} 
\usepackage{amsfonts}
\usepackage{amsthm}
\usepackage{amsmath}
\usepackage{float}
\usepackage[operators,sets]{cryptocode}
\usepackage[margin=1in]{geometry}

\theoremstyle{plain}
\newtheorem{problemnormal}{Problem}

\newtheorem{conjecturenormal}{Conjecture}
\newtheorem{goal}{Goal}

\title{Deterministic Hardness of Approximation \\ For SVP in all Finite $\ell_p$ Norms}
\author{Isaac M Hair \thanks{\texttt{isaacmhair@gmail.com}} \\ UCSB, UCLA \and Amit Sahai \thanks{\texttt{sahai@cs.ucla.edu}} \\ UCLA}

\begin{document}

\maketitle
\noteswarning

\thispagestyle{empty}

\begin{abstract}
    We show that, assuming NP $\not\subseteq$ $\cap_{\delta > 0}$DTIME$\left(\exp{n^\delta}\right)$, the shortest vector problem for lattices of rank $n$ in any finite $\ell_p$ norm is hard to approximate within a factor of $2^{(\log n)^{1 - o(1)}}$, via a deterministic reduction. Previously, for the Euclidean case $p=2$, even hardness of the \emph{exact} shortest vector problem was not known under a deterministic reduction.
\end{abstract}

\newpage

\setcounter{page}{1}

\section{Introduction}

Given an integer matrix $\mat B \in \mathbb{Z}^{M \times N}$, the lattice $\mathcal{L}(\mat B)$ generated by $\mat B$ is the set of all integral linear combinations of its rows,\footnote{All vectors throughout the paper are row vectors, unless stated otherwise.} i.e.
\[\mathcal{L}(\mat B) = \{\mat x \mat B : \mat x \in \mathbb{Z}^{M}\}.\]
Lattices have a wide range of applications in computer science, playing an important role in algorithmic number theory \cite{lenstra1982factoring}, integer programming \cite{lenstra1983integer, kannan1987minkowski, frank1987application, schrijver1998theory}, coding theory \cite{forney1988coset, de2002some, zamir2014lattice}, and perhaps most famously, post-quantum cryptography \cite{ajtai1998shortest, nguyen2001two, micciancio2007worst, regev2009lattices, peikert2016decade}.

An attractive feature of many contemporary lattice-based cryptosystems (see~\cite{peikert2016decade} and the references therein) is that, while the cryptosystems themselves sample lattice problems from a particular distribution, breaking the cryptosystems is as hard as solving certain worst-case lattice problems. One of the most important among these is the gap shortest vector problem (GapSVP). As the name suggests, the task is to approximate the length of the shortest nonzero vector in a given lattice:

\begin{problemnormal}[$\gamma$-GapSVP$_p$]
    Given a matrix $\mat B \in \mathbb{Z}^{M \times N}$ and a threshold $h$, distinguish between the following two cases:\footnote{The $\ell_p$ norm of a vector $\mat w \in \mathbb{Z}^N$ is defined as $\|\mat w\|_p = \left(\sum_{1 \leq i \leq N}\vert \mat w_i\vert^p\right)^{1/p}$. When $p = 0$, we instead define the $\ell_0$ pseudo-norm $\|\mat w\|_0$ as the number of nonzero entries in $\mat w$, and when $p = \infty$, we define $\|\mat w\|_\infty = \max_{1 \leq i \leq N}\vert\mat w_i\vert$.}
    \begin{enumerate}
        \item There exists a nonzero vector $\mat x \in \mathbb{Z}^M$ such that $\|\mat x\mat B\|_p \leq h$.
        \item For all nonzero vectors $\mat x \in \mathbb{Z}^M$, we have $\|\mat x \mat B\|_p > \gamma h$.
    \end{enumerate}
    We assume that the rows of $\mat B$ are linearly independent; otherwise the problem is trivial. The exact version (where $\gamma = 1$) is referred to as SVP$_p$.
\end{problemnormal}

For cryptographic applications, the most relevant is GapSVP$_2$, i.e. the version of the problem where we measure vector length using the Euclidian $\ell_2$ norm. Even for GapSVP in other $\ell_p$ norms, there is a rich body of work giving upper bounds \cite{kannan1987minkowski, ajtai2002sampling, blomer2009sampling, eisenbrand2022approximate} and lower bounds \cite{van1981another, arora1997hardness, micciancio2001shortest, dinur2002approximating, haviv2007tensor, aggarwal2018gap, bhattiprolu2024inapproximability, hair2025svp, hecht2025deterministic}.

Van Emde Boas \cite{van1981another} gave one of the first hardness results for the shortest vector problem, showing that SVP$_\infty$ is NP hard. Arora, Babai, Stern, and Sweedyk \cite{arora1997hardness} showed NP hardness of approximation for GapSVP$_\infty$, which was subsequently improved by Dinur \cite{dinur2002approximating}, culminating in a proof that $\gamma$-GapSVP$_\infty$ is NP hard to approximate within a nearly polynomial factor of $\gamma = n^{\Omega(1/\log\log n)}$.

The situation for finite $p$ is more nuanced. A breakthrough result by Ajtai \cite{ajtai1998shortest} showed that SVP$_2$ is hard, but only under the assumption that NP $\not\subseteq$ RP; his reduction makes critical use of randomization. Over the next few years, researchers improved upon Ajtai's results \cite{cai1998approximating, micciancio2001shortest, khot2003hardness, khot2005hardness, haviv2007tensor}, showing hardness of GapSVP for successively stronger approximation factors, but \emph{without removing the need for randomization}. The culmination of this line of work is summarized in the theorem below:

\begin{theorem}[Due to \cite{khot2005hardness, haviv2007tensor}]
    For all finite $p \geq 1$, $\gamma$-GapSVP$_p$ on lattices of rank $n$ is hard
    \begin{enumerate}
        \item when $\gamma$ is any constant, assuming NP $\not\subseteq$ RP.
        \item when $\gamma = 2^{(\log n)^{1 - \varepsilon}}$ for any constant $\varepsilon > 0$, assuming NP $\not\subseteq \text{RTIME}(2^{(\log n)^{O(1)}})$.
        \item when $\gamma = n^{c/\log\log n}$ for some constant $c > 0$, assuming NP $\not\subseteq \cap_{\delta > 0}\text{RTIME}(\exp(n^\delta))$.
    \end{enumerate}
\end{theorem}

This was the state of the art until very recently, when two independent papers by Hair and Sahai \cite{hair2025svp} and Hecht and Safra \cite{hecht2025deterministic} de-randomized and improved some of these results whenever $p > 2$. In particular, Hair and Sahai showed that for all constants $p > 2$, $\gamma$-GapSVP$_p$ is NP hard (under a deterministic reduction) to approximate within a factor of $\gamma = 2^{(\log n)^{1 - \varepsilon}}$, where $\varepsilon > 0$ is any constant.

\paragraph{Deterministic Hardness for $\ell_2$?} Putting everything together, GapSVP is known to be NP hard for all $\ell_p$ norms with $p > 2$. But for the cryptographically relevant case of $p = 2$, all known lower bounds critically leverage randomization; even showing deterministic hardness of the exact version of the problem has been open for over 40 years.
Indeed, showing hardness of GapSVP in the $\ell_2$ norm under a deterministic reduction is often described as an outstanding open problem; see for example \cite{van1981another, micciancio2001shortest, haviv2007tensor,micciancio2012inapproximability, micciancio2014locally, bennett2022hardness, bennett2023complexity, bennett2023parameterized}. In the words of Bennett, Peikert, and Tang, ``derandomizing hardness reductions for SVP (and similarly, BDD) is a notorious, decades-old open problem'' \cite{bennett2021improved}. 
One very recent mention is in a paper by Hecht and Safra \cite{hecht2025deterministic}, in which they ask for a proof of the following conjecture:
\begin{conjecturenormal}[Conjecture 7.1 from \cite{hecht2025deterministic}] \label{conj:safra}
    ``Under deterministic reductions, GapSVP is hard to approximate in the $\ell_2$ norm [...]''
\end{conjecturenormal}

Unfortunately, it does not appear that  existing techniques are well-equipped to resolve this conjecture. In particular, all existing reductions  either leverage a gadget called a \emph{locally dense lattice}, for which there are no known deterministic constructions, or else leverage what we refer to as \emph{anti-concentration gadgets}, which are only known to exist\footnote{That is, the barrier is existential, not just a de-randomization issue.} for $\ell_p$ norms with $p > 2$. We now elaborate.

\paragraph{A Barrier: Locally Dense Lattices.}
Ajtai's breakthrough hardness result for SVP$_2$ \cite{ajtai1998shortest}, and the reductions that followed \cite{cai1998approximating, micciancio2001shortest, khot2003hardness, khot2005hardness, haviv2007tensor}, all have roughly the same structure. The first step is to show hardness of approximation for (different variants of) the \emph{closest vector problem}:
\begin{problemnormal}[$\gamma$-GapCVP$_p$]
    Given a matrix $\mat B \in \mathbb{Z}^{M \times N}$, \emph{a vector $\mat w \in \mathbb{Z}^N$,} and a threshold $h$, distinguish between the following two cases:
    \begin{enumerate}
        \item There exists a vector $\mat x \in \mathbb{Z}^M$ such that $\|\mat x\mat B - \mat w\|_p \leq h$.
        \item For all vectors $\mat x \in \mathbb{Z}^M$, we have $\|\mat x \mat B - \mat w\|_p > \gamma h$.
    \end{enumerate}
\end{problemnormal}

In other words, the reductions start by showing hardness of approximation for an \emph{inhomogeneous} version of GapSVP, where instead of finding a short vector, the goal is to find a vector close to some target vector.

The main technical work lies in reducing from GapCVP to GapSVP. This homogenization process is performed by using a special gadget called a \emph{locally dense lattice}. All of the reductions use a randomized procedure to construct these gadgets, which means that the final lower bound for GapSVP only holds assuming NP $\not\subseteq$ RP.

For multiple decades, researchers have attempted to find a deterministic construction of locally dense lattices. Micciancio \cite{micciancio2001shortest} gave a deterministic construction under a certain number theoretic conjecture on the distribution of square-free smooth numbers. A few years later, the same author \cite{micciancio2012inapproximability} gave a different construction that, while still randomized, ensures that the final reduction to GapSVP will only have one sided error. Micciancio referred to this as a ``partial de-randomization.'' Later papers \cite{micciancio2014locally, bennett2022hardness, bennett2023complexity} all explore different approaches for de-randomization. But as of today, none of these approaches have resulted in a deterministic hardness result for GapSVP.

\paragraph{Another Barrier: Anti-Concentration Gadgets.}
Dinur's result for GapSVP$_\infty$ \cite{dinur2002approximating}, and the recent papers \cite{hair2025svp, hecht2025deterministic} targeted at $\ell_p$ norms with $p > 2$, follow a different approach. The reductions start by showing hardness of approximation for various \emph{constraint satisfaction problems} (CSPs). While the exact properties of these CSPs differ significantly between  papers, they all consist of a set of \emph{variables}, a set of \emph{constraints}, and an \emph{alphabet}, and the task is always to satisfy the constraints by assigning alphabet symbols to  variables.

To reduce from a CSP to GapSVP, the idea is to represent the CSP as a matrix, and then insert various gadgets into this matrix to get the final set of basis vectors. Each constraint from the CSP maps to an entire \emph{set of basis vectors}, with each individual basis vector representing a locally satisfying assignment for the constraint. The hope is to enforce that there exists a short linear combination of basis vectors if and only if the starting CSP was satisfiable.

A difficulty in the reductions is that an adversary is free to choose \emph{any} linear combination of basis vectors to make a lattice vector. In particular, the adversary might choose basis vectors that only correspond to a small number of constraints from the CSP, in which case it is difficult to implicate the properties of the CSP. To fix this issue, the reductions all use something we refer to as a (deterministic) \emph{anti-concentration gadget}.

Denoting this matrix gadget as $\mat Y$, and the initial set of basis vectors as $\mat B$, the final set of basis vectors will look something like $\big[\mat B \| \mat Y\big]$. Roughly speaking, the gadget $\mat Y$ is aligned with the basis matrix $\mat B$ in such a way that (i) linear combinations of basis vectors that correspond to only a few constraints of the CSP will map to long lattice vectors, and (ii) linear combinations which are more spread out across the constraints map to short lattice vectors. Unfortunately, \emph{such gadgets are only known to exist for $\ell_p$ norms with $p > 2$.}

\paragraph{Our Results.} Despite these barriers, we give the first deterministic hardness of approximation result for GapSVP in the $\ell_2$ norm; see Theorem \ref{thm:main} for a formal statement.

\begin{theorem}[Informal] \label{thm:maininformal}
    Let $p \geq 1$ be any constant. Assuming NP $\not\subseteq  \cap_{\delta > 0}$DTIME$\left(\exp{n^\delta}\right)$, then $\gamma$-GapSVP$_p$ on lattices of rank $n$ is hard to approximate within a factor of $\gamma = 
    2^{(\log n)^{1-\varepsilon}}
    $ where \[\varepsilon = \frac{(\log\log\log n)^{O(1)}}{\sqrt{\log\log n}} = o(1).\]
\end{theorem}

This resolves Conjecture \ref{conj:safra} in the affirmative.  Note that our result is also the first to show hardness of \emph{exact} SVP$_2$ under a deterministic reduction. 
Our hardness reduction departs significantly from existing lower bounds for the shortest vector problem over $\ell_2$. We do not make use of locally dense lattices, and we do not give a reduction by way of the closest vector problem. Instead, we draw inspiration from the line of work targeting GapSVP$_p$ in higher $\ell_p$ norms \cite{dinur2002approximating, hair2025svp, hecht2025deterministic}, although as we discuss below, our technical approach is quite different.

Our main technical contribution is a new type of tensor product, dubbed the \emph{Vandermonde fortified tensor product}, which we apply to the matrix representation of a quasi-random PCP~\cite{khot2006ruling}. The purpose of this tensor product is to amplify certain properties of the quasi-random PCP, so that the resulting matrix {is} amenable to a gadgetized reduction to GapSVP. We emphasize that Vandermonde fortified tensor products serve a critical role in the reduction, instead of being used to boost the approximation factor as in~\cite{khot2005hardness, haviv2007tensor}. Indeed, even if our goal was just to show deterministic hardness of \emph{exact} SVP$_2$, our  approach still requires the use of these fortified tensor products.

\section{Technical Overview}
\label{sec:tech}

In this overview, our focus will be to show constant factor hardness of approximation for GapSVP$_2$ via a deterministic reduction from SAT. We will in fact reduce to a version of the problem where satisfiable SAT instances map to short lattice vectors that also have all of their entries in $\{-1,0,1\}$.

With this (and some additional requirements) at hand, the approximation factor can be amplified via direct tensoring; after balancing all of the relevant parameters, we recover Theorem \ref{thm:maininformal}. For a more detailed discussion of the direct tensoring step, see Section \ref{sec:setup}.

\subsection{A Different Starting Problem} \label{sec:newprob}
Towards finding a reduction that works for the $\ell_2$ norm (and indeed all $\ell_p$ norms), we will not start from the closest vector problem (as was done by \cite{ajtai1998shortest, cai1998approximating, micciancio2001shortest, khot2003hardness, khot2005hardness, haviv2007tensor}), and we will not start from a constraint satisfaction problem (as was done by \cite{dinur2002approximating, hecht2025deterministic, hair2025svp}). Instead, our reduction begins with an instance of the hypergraph problem stated below (see Problem \ref{prob:qrdh} for a formal statement). As discussed in Section \ref{sec:qrdhproof}, this problem is an equivalent formulation of the inner verifier in Khot's quasi-random PCP \cite{khot2006ruling}, which was introduced in the context of proving hardness of approximation for various graph problems including graph min-bisection, densest $k$-subgraph, and bipartite clique.

\begin{problemnormal}[Quasi-Random Densest Sub-Hypergraph (QRDH), Informal] \label{prob:informalqrdh}
    Given a hypergraph $H = (V, E)$ of arity $d$ with $\vert V\vert = N$ and $\vert E\vert = M$, along with parameters $r \geq 1$ and $\beta \in [0, 1]$, distinguish between the following two cases:
    \begin{enumerate}
        \item There exists a vertex subset $V' \subseteq V$ of size at most $N/r$ that fully contains at least $(1/r)^{d-1} M$ hyperedges.
        \item Every vertex subset $V' \subseteq V$ fully contains at most $(\vert V'\vert/N)^d M + \beta M$ hyperedges.
    \end{enumerate}
\end{problemnormal}

QRDH is very similar to the usual densest sub-hypergraph problem. To see the connection, assume that $\beta$ is relatively small, fix a vertex subset size parameter $k = N/r$, and a hyperedge subset size parameter $h = (1/r)^{d-1} M$. Then in case (1), there exists a size-$k$ vertex subset that contains $h$ hyperedges, whereas in case (2) every size-$k$ vertex subset contains at best slightly more than $h/r$ hyperedges.

The difference between QRDH and densest sub-hypergraph is that case (2) imposes a stronger requirement. We want that \emph{every} vertex subset $V'$ of \emph{any} size contains (at best) only slightly more than $(\vert V'\vert/N)^d M$ hyperedges. Notice that no matter the structure of the hypergraph, if we choose a subset $V'$ at random, we expect it to fully contain approximately $(\vert V'\vert/N)^d M$ hyperedges. This means that the upper bound in case (2) is essentially the best we can hope for. It turns out that random hypergraphs satisfy case (2) with high probability, which is why we follow~\cite{khot2006ruling} and refer to the problem as ``quasi-random.'' This property will play an important role later, when we manipulate instances of QRDH using a special tensor product.

We show in Section \ref{sec:qrdhproof} that a relatively direct adaptation of Khot's original quasi-random PCP \cite{khot2006ruling} and the quasi-random PCP presented by Khot and Saket \cite{khot2016hardness} gives us 
the following theorem (we omit factors of $\log\log\log M$; see Theorem \ref{thm:khotqrdh} for a formal statement):

\begin{theorem}[Informal] \label{thm:informalkhotqrdh}
    There is a deterministic, subexponential time reduction from SAT to
    \[(N \leq M^{O(1)}, M, d = \log\log M, r = \log M, \beta = 1/(\log M)^{100\log\log M})\text{-QRDH}.\]

\end{theorem}

Let hypergraph $H = (V, E)$ be the QRDH instance output by the reduction, and set $k = N/r = N/\log M$ and $h = (1/r)^{d-1}M = M/(\log M)^{\log\log M - 1}$. Cases (1) and (2) of the QRDH problem imply the following \emph{vertex-hyperedge containment} gap for $H$:
\begin{enumerate}
    \item If the starting SAT instance was satisfiable, then there exists a size-$k$ vertex subset that fully contains at least $h$ hyperedges.
    \item Otherwise, every size-$k$ vertex subset fully contains at most $\frac{1 + o(1)}{\log M} \cdot h$ hyperedges.
\end{enumerate}

This gives nearly a $\log M$ factor gap between cases (1) and (2). But for our reduction to GapSVP, we will not be concerned with the number of edges contained within a size-$k$ vertex subset. Instead, we view the problem in a contrapositive form, and ask: what is the minimum number of vertices touched by any size-$h$ subset of hyperedges? From this perspective, the dichotomy between cases (1) and (2) is weaker, giving only a constant factor \emph{hyperedge-vertex expansion} gap:\footnote{For technical reasons, our reduction actually requires the initial hyperedge-vertex expansion to be at least $d^{\omega(1)} = (\log\log M)^{\omega(1)}$, not just constant. The disparity arises because we are ignoring the factors of $\log\log\log M$ in Theorem \ref{thm:informalkhotqrdh}. For the sake of exposition, we ignore this detail for the rest of the overview.}
\begin{enumerate}
    \item If the starting SAT instance was satisfiable, then there exists a size-$h$ subset of hyperedges touching at most $k$ vertices.
    \item Otherwise, every size-$h$ subset of hyperedges touches at least $(2-o(1)) \cdot k$ vertices.
\end{enumerate}
To permit a reduction to GapSVP$_2$, we need to amplify this hyperedge-vertex expansion gap from $(2-o(1))$ to at least\footnote{The constant $100$ in the exponent here is only for illustrative purposes.} $(\log M)^{100}$. We note that, even if the goal was just to show hardness of exact SVP$_2$, our general technique still requires this amplification step.

The overall plan for our reduction will go as follows. In Section \ref{sec:newtensor}, we pre-process the QRDH instance using a special type of tensor product, which allows us to achieve the desired hyperedge-vertex expansion gap, albeit in a slightly different form. This transformation is the main technical contribution of our paper. Then in Section \ref{sec:puttingtogether}, we manipulate the resulting matrix, inserting carefully chosen gadgets to get our final set of basis vectors.

\subsection{A New Tensor Product} \label{sec:newtensor}
In this section, we introduce the \emph{Vandermonde fortified tensor product} (VF tensor product). Towards reasoning about tensor products, view the hypergraph $H = (V, E)$ output by Theorem \ref{thm:informalkhotqrdh} in terms of its 0/1 indicator matrix $\mat P$. Each column of $\mat P$ is indexed by a distinct vertex of $H$, each row of $\mat P$ is indexed by a distinct hyperedge of $H$, and every row has $d$ nonzero entries to indicate the vertices contained in the corresponding hyperedge. We say that a subset $V' \subseteq V$ of columns in $\mat P$ (vertices in $H$) \emph{supports} a row $e$ of $\mat P$ (hyperedge $e$ of $H$) whenever all the nonzero entries of row $e$ fall inside of $V'$ (the vertices of $e$ are all contained in $V'$).

Roughly speaking, the $q$-fold VF tensor product of $\mat P$ is a matrix $\mat T = \big[\mat A \| \mat Q \big]$, where $\mat Q = \mat P^{\otimes q}$ is the $q$-fold tensor product of $\mat P$, and $\mat A$ is a special gadget based on Vandermonde matrices. (We sketch the construction later in this subsection. See Definition \ref{defn:VFTP} for a formal definition of VF tensor products.) Each row of $\mat Q$ corresponds to a $q$-tuple of rows from $\mat P$. As such, we index each row of $\mat Q$ using the corresponding $q$-tuple $(e_1, \ldots, e_q) \in [M]^q$ of row indices for $\mat P$, or equivalently the corresponding $q$-tuple $(e_1, \ldots, e_q) \in [M]^q$ of hyperedges from $H$. We use the same row indexing scheme for $\mat T$. Sometimes, we refer to the rows of $\mat Q$ as hyperedges, since $\mat Q$ is the indicator matrix for the $q$-fold tensor product of the hypergraph $H$ with itself.

\paragraph{Why are VF Tensor Products Useful?}
Before discussing VF tensor products in more detail, we characterize exactly what properties we want to achieve, and relate these properties back to our overall reduction to GapSVP$_2$. Recall that our goal from Section \ref{sec:newprob} was to amplify the hyperedge-vertex expansion gap for $H$, and that we start with a multiplicative gap of $(2 - o(1))$. As we discuss later in this section, using the tensor product $\mat Q = \mat P^{\otimes q}$ itself does \emph{not} give us the desired gap amplification, because we cannot rule out the existence of certain \emph{ill-behaved hyperedge subsets}, i.e. hyperedge subsets that are supported on a vertex subset which is too small.

Our insight is that \emph{every ill-behaved hyperedge subset has special structure}. 
This allows us to construct a matrix $\mat A$ with the following property: Let $\mat x \in \mathbb{Z}^{M^q}$ be any nonzero coefficient vector, and let $E' \subseteq [M]^q$ be the indices for the nonzero entries of $\mat x$. Then if $E'$ is ill-behaved, it must be the case that $\mat x \mat A \neq \mat 0$.


\begin{remark}
At first glance, it may seem that we are simply requiring the rows of $\mat A$ indexed by $E'$ to be linearly independent. However, there is a subtle but crucial distinction:
Notice that by construction $\mat x$ will have a nonzero entry corresponding to \emph{every} index in $E'$. This means that $E'$ might index a linearly dependent set of rows, but at the same time every linear combination which gives a nonzero coefficient to \emph{all} of those rows results in a nonzero sum.
\end{remark}

Looking ahead in our reduction to GapSVP$_2$, each row of $\mat T = \big[\mat A \| \mat Q\big]$ will map to a basis vector, and we will ensure that every subset of basis vectors that does \emph{not} admit a zero-sum linear combination of the rows in $\mat A$ will always map to a long lattice vector. That way, the only subsets of basis vectors which are even relevant for constructing short lattice vectors \emph{must correspond to well-behaved hyperedge subsets.}

As a result, we know that every short lattice vector must come from a linear combination of basis vectors that is well-behaved. And, there is a large hyperedge-vertex expansion gap for all the well-behaved subsets of rows in $\mat Q$, depending on whether the starting SAT instance was satisfiable or not. The only remaining detail is that, of course, we would like satisfiable SAT instances to map to short lattice vectors.

In the context of our final reduction, we will need satisfiable SAT instances to map to a set of rows in $\mat T$, and hence a set of basis vectors, which is not only well-behaved, but also admits a zero-sum linear combination of the corresponding rows of $\mat A$ \emph{that has small coefficients}. A slightly stronger property that implies this is the following: If the starting SAT instance was satisfiable, then there exists a subset $E' \subseteq [M]^q$ that (i) maps to a compressing set of rows in $\mat Q$, and (ii) maps to a set of rows in $\mat A$ whose nonzero entries are supported\footnote{The constant $100$ in the exponent here is only for illustrative purposes.} on at most $\vert E'\vert/(\log M)^{100}$ columns. 

Below, we summarize all of these requirements. 

\begin{goal} Let $H$ be the hypergraph output by the reduction in Theorem \ref{thm:informalkhotqrdh}, and denote by $\mat P$ its indicator matrix. Set $k = N/r = N/\log M$ and $h = (1/r)^{d-1}M = M/(\log M)^{\log\log M - 1}$. The desired properties of the $q$-fold VF tensor product $\mat T = \big[\mat A \| \mat Q \big]$ of $\mat P$ are: \label{goal:VFTP}
\begin{enumerate}
    \item Suppose that the starting SAT instance was satisfiable. Then there exists a subset $E' \subseteq [M]^q$ of size $\vert E'\vert = h^q$ such that \emph{both of the following are satisfied}:
    \begin{enumerate}
        \item The rows indexed by $E'$ in $\mat Q$ are supported on at most $k^q$ columns.
        \item The rows indexed by $E'$ in $\mat A$ have their nonzero entries supported on at most $h^q/(\log M)^{100}$ distinct columns.
    \end{enumerate}
    \item Suppose that the starting SAT instance was unsatisfiable. Then for all subsets $E' \subseteq [M]^q$ of size $\vert E'\vert = h^q$, \emph{at least one of the following is satisfied}\footnote{Formally speaking, case (2) needs to hold for a range of subset sizes, but we restrict our attention to $\vert E'\vert = h^q$ for simplicity.}
    \begin{enumerate}
        \item The rows indexed by $E'$ in $\mat Q$ are supported on at least $(2-o(1))^qk^q$ columns.
        \item Every coefficient vector $\mat x \in \mathbb{Z}^{M^q}$ whose nonzero entries are exactly the entries indexed by $E'$ satisfies $\mat x \mat A \neq \mat 0$.
    \end{enumerate}
\end{enumerate}
\end{goal}

\paragraph{Shortcomings of Direct Tensoring.} A tantalizing approach to achieve the hyperedge-vertex expansion goal is to simply use the tensor product $\mat Q = \mat P^{\otimes q}$. We begin our discussion of VF tensor products by highlighting the shortcomings of this approach. Along the way, we highlight why we started with the \emph{quasi-random} densest subhypergraph problem, as opposed to the usual densest subhypergraph problem.

We re-state the hyperedge-vertex expansion gap for $H$ (without the ``quasi-random'' strengthening):
\begin{enumerate}
    \item If the starting SAT instance was satisfiable, then there exists a size-$h$ subset of hyperedges touching at most $k$ vertices.
    \item Otherwise, every size-$h$ subset of hyperedges touches at least $(2-o(1)) \cdot k$ vertices.
\end{enumerate}

Now suppose that we defined the $q$-fold VF tensor product of $\mat P$ simply as $\big[\mat 0 \| \mat Q\big]$, that is, we set the matrix $\mat A$ to just be zero. Conditions (1a) and (1b) of Goal \ref{goal:VFTP} are satisfied by definition of $\mat Q$ and by the fact that $\mat A = \mat 0$, respectively. But because all rows of $\mat A$ are linearly dependent, the only way for condition (2) to be satisfied is if we satisfy condition (2a). This could fail miserably, even in the case of $q = 2$. Suppose that $\mat P$ has a subset of $h/\log M$ rows supported on just $k^{0.1}$ columns, and a subset of $h \log M$ rows supported on $100k$ columns. Taking the tensor product of these two subsets, we get a subset $E' \subseteq [M]^2$ of size $h^2$ that indexes a set of rows whose nonzero entries are only supported on just $100k^{1.1}$ columns. This not only violates condition (2a) of Goal~\ref{goal:VFTP}, but is even more compressing than the guarantee in condition (1a)!

In this case, tensoring failed because our original conditions on the matrix $\mat P$ only mentioned row subsets of size $h$. However, because our starting hypergraph problem was quasi-random, we can enforce something useful about row subsets of all sizes.

Suppose that the starting SAT instance was unsatisfiable, so that $H$ satisfies case (2) of the QRDH problem. Then we know that every vertex subset of size $k'$ supports at most
\[(k'/N)^dM + \beta M = (k'/N)^{\log\log M}M + M/(\log M)^{100\log\log M}\]
hyperedges. We would like to find a contrapositive form of this statement. Notice that we have no guarantees on hyperedge subsets of size at most $\beta M$. 
But for a cutoff size of, say, $\beta^{1/3}M$, it turns out that every hyperedge subset of size $h' > \beta^{1/3}M$ is incident to at least
\begin{align} \label{eq:techlower}
    (1 - \beta^{1/3})(h'/M)^{1/d}N
\end{align}
vertices. Observe that in a random hypergraph of arity $d$, the lower bound would be approximatley $(h'/M)^{1/d}N$, so the above lower bound comes very close to the best possible.

Relating this back to condition (2a) of Goal \ref{goal:VFTP}, we would like to understand which subsets $E' \subseteq [M]^q$ of size $\vert E'\vert = h^q$ violate the condition, and which ones do not. For the time being, we limit ourselves to subsets $E' = E_1 \times \ldots \times E_q$ formed as a Cartesian product of factors $E_1, \ldots, E_q \subseteq [M]$. (We show implicitly in Theorem \ref{thm:vandtensor} that sets of this form are the worst case.) Suppose that all of the factors $E_i$ satisfy $\vert E_i\vert > \beta^{1/3}M$, so that the lower bound (\ref{eq:techlower}) applies. We argue that in this case, $E'$ satisfies condition (2a).

For all $i$, denote by $V_i$ the set of columns supporting the rows of $\mat P$ indexed by $E_i$. We know that $\vert E_1\vert \cdot \ldots \cdot \vert E_q\vert = h^q$ by assumption. By the Cartesian product structure, we know that the number of columns supporting the rows of $\mat Q$ indexed by $E'$ is equal to $\vert V_1\vert \cdot \ldots \cdot \vert V_q\vert$. Plugging in (\ref{eq:techlower}), and assuming $\beta^{1/3} \ll 1/q$, we have
\begin{align*}
    \prod_{1 \leq i \leq q}{\vert V_i\vert} & \geq (1 - \beta^{1/3})^q\prod_{1 \leq i \leq q}\left((\vert E_i\vert/M)^{1/d}N\right) \tag{Using (\ref{eq:techlower})} \\
    & \geq (1-o(1)) \prod_{1 \leq i \leq q}\frac{\vert E_i\vert^{1/d}N}{M^{1/d}} \tag{$\beta^{1/3} \ll 1/q$}\\
    & \geq (1-o(1)) \frac{h^{q/d}N^q}{M^{q/d}} \tag{$\prod_{1 \leq i \leq q}{\vert E_i\vert} = h^q$}\\
    & \geq (1-o(1)) (1/r)^{q - q/d}N^q \tag{$h = (1/r)^{d-1}M$}\\
    & \geq (1-o(1)) k^q(1/r)^{-q/d} \tag{$k = N/r$}\\
    & \geq (1-o(1)) 2^qk^q \tag{$r = \log M$ and $d = \log\log M$}
\end{align*}

In other words, we know that (at least when restricted to Cartesian product structure), the only subsets $E' \subseteq [M]^q$ that could possibly violate condition (2a) of Goal \ref{goal:VFTP} are those with at least one factor $E_i \subseteq [M]$ which is ``very small.'' These are exactly the ill-behaved subsets that the matrix $\mat A$ in our Vandermonde fortified tensor product is designed to filter out. (While this sketch only considers Cartesian product sets, we stress that our formal Theorem \ref{thm:vandtensor}  shows that, by including $\mat A$, we in fact filter out all possible ill-behaved subsets $E'$, not just those limited to Cartesian product structure.)

\paragraph{Constructing the VF Tensor Product.}
Recall that the $q$-fold VF tensor product of $H$'s indicator matrix $\mat P$ is denoted as $\mat T = \big[\mat A \| \mat Q\big]$, and that we set $\mat Q = \mat P^{\otimes q}$. All that remains is to specify the structure of the matrix $\mat A$, and describe how it will filter out ill-behaved subsets. Before doing so, we define a specific type of matrix over the integers that has strong linear independence properties:

\begin{definition}[Reduced Vandermonde Matrix]
    Let $a, b$ be positive integers such that $a > b$ and $a$ is prime. We define an $(a, b)$ reduced Vandermonde matrix $\mat V \in \mathbb{Z}^{(a - 1) \times b}$ as $\mat V_{i, j} = i^{j-1} \mod a$.
\end{definition}

We observe in Lemma \ref{lem:vandindependent} that every subset of at most $b$ rows from an $(a, b)$ reduced Vandermonde matrix is linearly independent (this follows by a simple application of known results on Vandermonde matrices).

Towards constructing $\mat A$, notice that the ill-behaved subsets $E' \subseteq [M]^q$ with Cartesian product structure all share a common feature: because one of the factors $E_i \subseteq [M]$ is small (of size at most $\beta^{1/3}M$ but at least $1$), we know that there exists a choice of coordinate $\ell \in [q]$, and a choice of indices $e_1, \ldots, e_{\ell-1}, e_{\ell+1}, \ldots, e_q \in [M]$ for the other coordinates, such that the set
$S = \{(e_1, \ldots, e_{\ell-1}, e_\ell', e_{\ell+1}, \ldots, e_q) : e_\ell' \in [M]\}$
satisfies
\[0 < \vert E' \cap S\vert \leq \beta^{1/3}M.\]

Using this observation, the construction of $\mat A$ is simple. For every possible set $S$, we append a matrix $\mat A^{(S)} \in \mathbb{Z}^{M^q \times \beta^{1/3}M}$ to $\mat A$. Each row of $\mat A^{(S)}$ indexed by a member of $S$ is set to a distinct row of a width-$\beta^{1/3}M$ reduced Vandermonde matrix, and all the other rows are set to zero.

Notice that for each $S$, we have the following properties:
\begin{enumerate}
    \item Every subset $E' \subseteq [M]^q$ satisfying $0 < \vert E' \cap S\vert \leq \beta^{1/3}M$ indexes a linearly independent set of rows in $\mat A^{(S)}$.
    \item All other subsets index a linearly dependent set of rows in $\mat A^{(S)}$.
\end{enumerate}

\begin{figure}[H]
    \centering
    \includegraphics[width=0.8\textwidth]{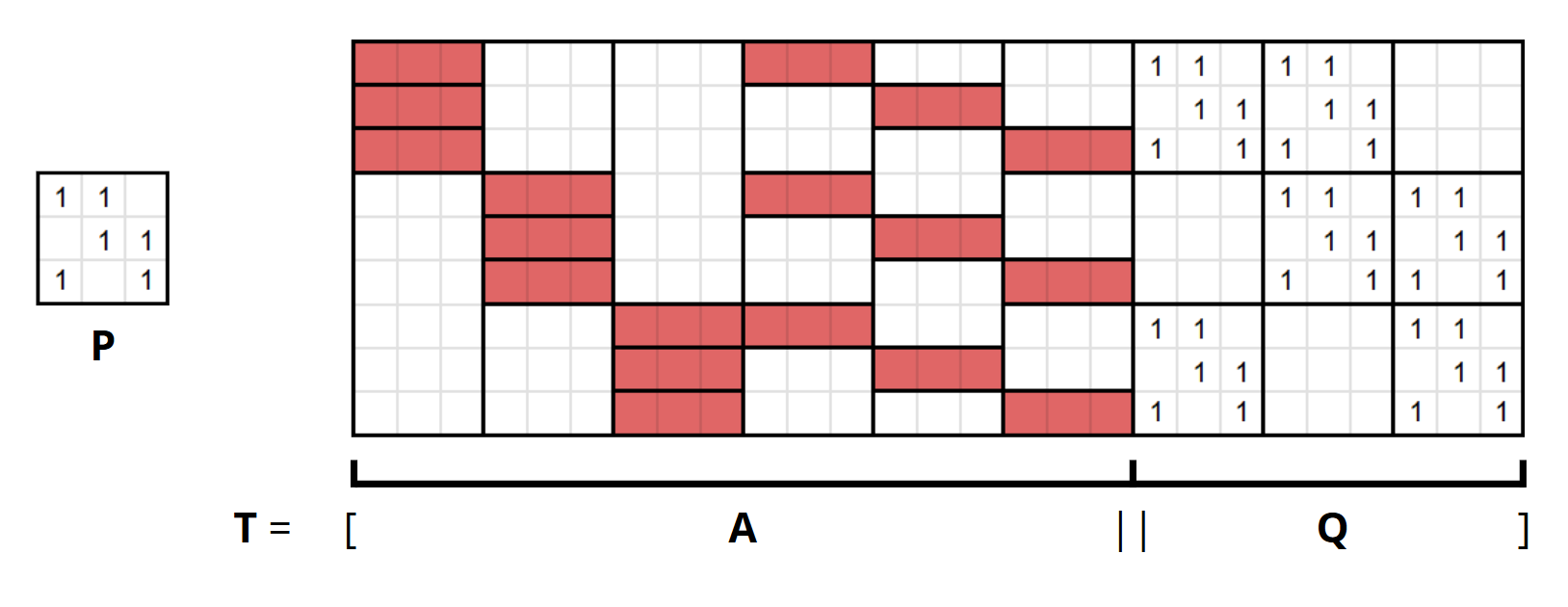}
    \caption{(Left) An example 0/1 matrix $\mat P$ with $N = 3$ columns and $M = 3$ rows. (Right) The $2$-fold VF tensor product of $\mat P$, denoted as $\mat T = \big[\mat A \| \mat Q\big]$. Each (empty) white square is a zero entry, and each red bar is a row taken from a reduced Vandermonde matrix.}
    \label{fig:VFTP}
\end{figure}

All that remains is to relate back to Goal \ref{goal:VFTP}. Suppose that the starting SAT instance was satisfiable, meaning there exists a set of $h$ rows in $\mat P$ supported on at most $k$ columns. Taking the $q$-fold tensor product of this set, we directly get a set of $h^q$ rows in $\mat Q$ supported on at most $k^q$ columns, meaning that we satisfy condition (1a). Let the indices for these rows be $E' \subseteq [M]^q$. Notice that for all submatrices $\mat A^{(S)}$, we either have $\vert E' \cap S\vert = 0$, or $\vert E' \cap S \vert = h \gg \beta^{1/3}M$. After some calculations, this means that (assuming $\beta^{1/3} \ll 1/q$) the row-induced submatrix $\mat A'$ of $\mat A$ indexed by $E'$ will have less than $h^q/(\log M)^{100}$ nonzero columns, meaning that we satisfy condition (1b).

If the SAT instance wasn't satisfiable, then every subset $E'$ is either well-behaved, and hence satisfies condition (2a), or else it is ill-behaved, and there exists at least one sub-matrix $\mat A^{(S)}$ of $\mat A$ such that $0 < \vert E' \cap S\vert \leq \beta^{1/3}M$. In this case, every nonzero coefficient vector $\mat x \in \mathbb{Z}^{M^q}$ whose nonzero entries are exactly the entries indexed by $E'$ will satisfy $\mat x \mat A^{(S)} \neq \mat 0$, and hence $\mat x \mat A \neq \mat 0$, meaning that we satisfy condition (2b).

\subsection{The GapSVP$_2$ Reduction} \label{sec:puttingtogether}
With the Vandermonde fortified tensor product in hand, the rest of the reduction proceeds similarly to the reduction by Hair and Sahai \cite{hair2025svp}, but without needing anti-concentration gadgets. We sketch the main steps below; see Section \ref{sec:constgap} for the details.

Apply Theorem \ref{thm:informalkhotqrdh} to the starting SAT instance to obtain a hypergraph $H$ with indicator matrix $\mat P$. Then we let $\mat T = \big[\mat A \| \mat Q\big]$ be the $q$-fold VF tensor product of $\mat P$, where\footnote{The value of $q$ in our actual reduction differs slightly from this by some $\log\log\log M$ factors.} $q = (\log\log M)^2$. As stated in Goal \ref{goal:VFTP}, we know the following:
\begin{enumerate}
    \item Suppose that the starting SAT instance was satisfiable. Then there exists a subset $E' \subseteq [M]^q$ of size $\vert E'\vert = h^q$ such that \emph{both of the following are satisfied}:
    \begin{enumerate}
        \item The rows indexed by $E'$ in $\mat Q$ are supported on at most $k^q$ columns.
        \item The rows indexed by $E'$ in $\mat A$ have their nonzero entries supported on at most $h^q/(\log M)^{100}$ distinct columns.
    \end{enumerate}
    \item Suppose that the starting SAT instance was unsatisfiable. Then for all subsets $E' \subseteq [M]^q$ of size $\vert E'\vert = h^q$, \emph{at least one of the following is satisfied}
    \begin{enumerate}
        \item The rows indexed by $E'$ in $\mat Q$ are supported on at least $(2-o(1))^qk^q = (\log M)^{\omega(1)}k^q$ columns.
        \item Every coefficient vector $\mat x \in \mathbb{Z}^{M^q}$ whose nonzero entries are exactly the entries indexed by $E'$ satisfies $\mat x \mat A \neq \mat 0$.
    \end{enumerate}
\end{enumerate}

Now, we replace each nonzero entry in $\mat Q$ with a distinct row of a width-$(h^q/k^q)/(\log M)^{100}$ reduced Vandermonde matrix, and denote by $\mat R$ the resulting matrix. (See Figure~\ref{fig:PQR} for an example.)

Suppose that the starting SAT instance was satisfiable. Let $E'$ be the subset guaranteed to exist by condition (1) from above. Then $E'$ indexes a set of $h^q$ rows in $\mat R$, whose nonzero entries are supported on at most $k^q \cdot (h^q/k^q)/(\log M)^{100} = h^q/(\log M)^{100}$ distinct columns of $\mat R$. We also know that $E'$ indexes a set of $h^q$ rows in $\mat A$, whose nonzero entries are supported on at most $h^q/(\log M)^{100}$ distinct columns. Thus the submatrix of $\big[\mat A \| \mat R\big]$ indexed by $E'$ is highly compressing.

Now suppose that the starting SAT instance was not satisfiable. Combining the properties of $\mat A$ with the properties of $\mat R$, we have the following. Let $\mat x \in \mathbb{Z}^{M^q}$ be any coefficient vector with $h^q$ nonzero entries, and let $E' \subseteq [M]^q$ be the set of all indices for these nonzero entries. Then either $\mat x \mat A \neq \mat 0$, or else the submatrix of $\mat R$ indexed by $E'$ has its nonzero entries supported on at least $(\log M)^{\omega(1)} k^q \cdot (h^q/k^q)/(\log M)^{100} = h^q(\log M)^{\omega(1)}$ distinct columns. In this second case, the submatrix of $\mat R$ indexed by $E'$ is highly expanding.

\begin{figure}[H]
    \centering
    \includegraphics[width=0.8\textwidth]{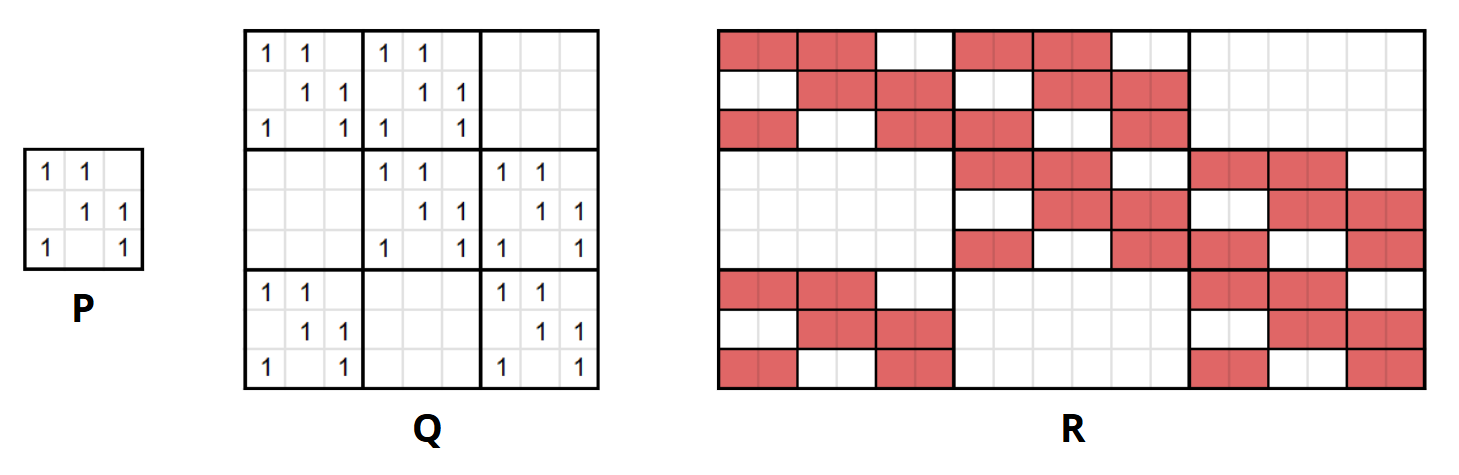}
    \caption{An example sequence of matrices $\mat P, \mat Q$, and $\mat R$. Each (empty) white square is a zero entry, and each red bar is a row taken from a reduced Vandermonde matrix.} 
    \label{fig:PQR}
\end{figure}

The last few steps are as follows: First, we append a width-$h^q/(\log M)^{100}$ reduced Vandermonde matrix $\mat W$ to $\big[\mat A \| \mat R\big]$, to get a matrix $\mat C = \big[\mat A \| \mat R \| \mat W\big]$. This enforces that any nontrivial linear combination of the rows of $\mat C$ which is even able to cancel out $\mat W$ must use at least $h^q/(\log M)^{100}$ of the rows. Combining this with the properties of $\big[\mat A \| \mat R\big]$, and performing several calculations (see Sections \ref{sec:complete} and \ref{sec:sound}), we can guarantee the following:
\begin{enumerate}
    \item Suppose that the starting SAT instance was satisfiable. Then there exists a nonzero vector $\mat x \in \{-1,0,1\}^{M^q}$ with at most $h^q$ nonzero entries, such that $\mat x \mat C = \mat 0$.
    \item Suppose that the starting SAT instance was unsatisfiable. Then for all nonzero vectors $\mat x \in \mathbb{Z}^{M^q}$ such that $\mat x \mat C = \mat 0$, it must be that $\mat x$ has at least $2h^q$ nonzero entries.
\end{enumerate}
In case (1), we have that $\mat x \in \{-1,0,1\}^{M^q}$, not just $\mat x \in \mathbb{Z}^{M^q}$, because in this case we can find a submatrix of $\mat C$ with $h^q$ rows that is compressing \emph{by a factor of $\Omega((\log M)^{100})$.} This allows us to use the pigeonhole principle to find a zero-sum linear combination with small coefficients.

\begin{figure}[H]
    \centering
    \includegraphics[width=0.8\textwidth]{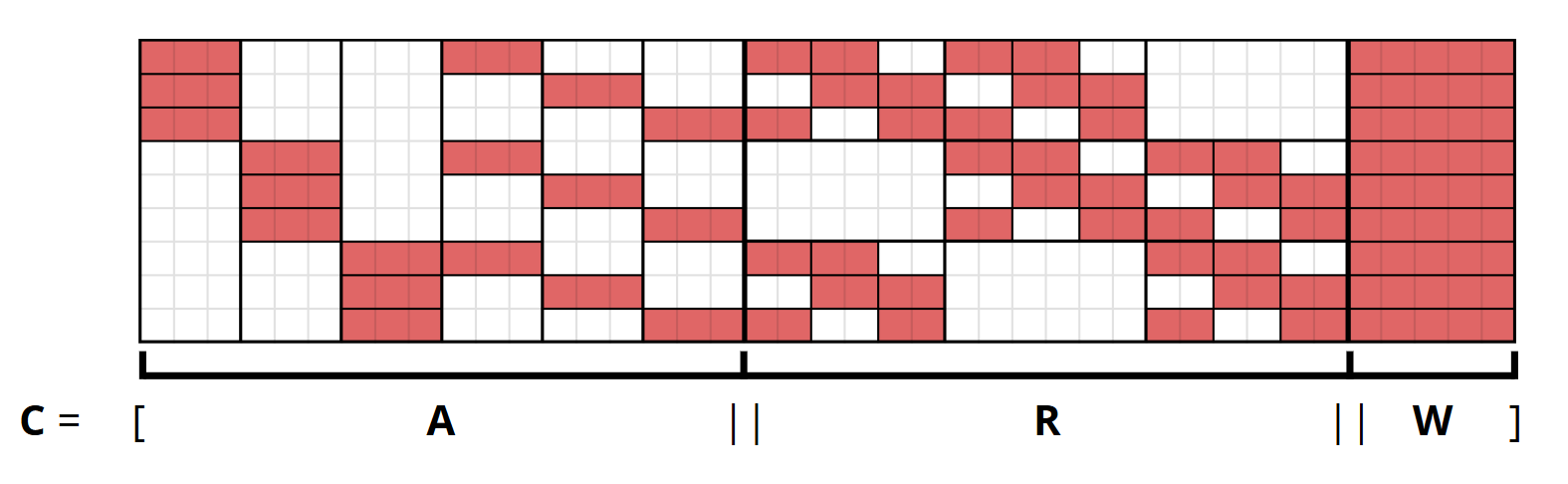}
    \caption{An example matrix $\mat C$. As in Figure \ref{fig:PQR}, each (empty) white square is a zero entry, and each red bar is a row taken from a reduced Vandermonde matrix.} 
    \label{fig:matC}
\end{figure}

To get our final basis matrix $\mat B$, we just horizontally concatenate $2h^q$ copies of $\mat C$, and then append a single $M^q \times M^q$ identity matrix. This ensures that, if the starting SAT instance was satisfiable, there exists a nonzero vector $\mat x$ such that $\mat x \mat B = \big[ \mat 0 \| \mat x \mat I\big]$ has at most $h^q$ nonzero entries, all of which are in $\{-1,1\}$. On the other hand, if the starting SAT instance was unsatisfiable, then every nonzero vector $\mat x$ either satisfies either (i) $\mat x \mat C \neq \mat 0$, in which case the $2h^q$ copies of $\mat C$ guarantee that $\mat x \mat B$ has at least $2h^q$ nonzero entries, or (ii) $\mat x$ itself has at least $2h^q$ nonzero entries, in which case the identity matrix in $\mat B$ ensures $\mat x \mat B$ has at least $2h^q$ nonzero entries.
In either case, we get the $\ell_2$ norm gap that we are aiming for.

\section{Setting Up the Main Reduction}\label{sec:setup}

All vectors throughout the paper are row vectors unless otherwise stated, and all logs/exponentials are taken base $2$ unless otherwise stated. We use $[n]$ to denote the set $\{1, \ldots, n\}$, and $[a, b]$ to denote the set $\{a, a+1, \ldots, b\}$. We index the rows and columns of an $m \times n$ matrix $\mat A$ using members of $[m]$ and $[n]$, respectively. $\mat A_i$ denotes the $i$\textsuperscript{th} row of $\mat A$, and $\mat A_{\cdot, j}$ denotes the $j$\textsuperscript{th} column of $\mat A$. We assume throughout the paper that all matrices have at least one row/column, and all hypergraphs have at least one hyperedge/vertex.

\paragraph{Formal Statement of the Main Theorem.}
Our overall goal is to prove the following theorem:

\begin{theorem}\label{thm:main}
    Let $p \geq 1$ be any constant. There is a deterministic $\exp{n^{O(1/\log\log n)}}$ time reduction from SAT instances of size $n$ to $\gamma$-GapSVP$_p$ on lattices of rank $M$, where $M = \exp{n^{O(1/\log\log n)}}$ and $\gamma = \exp{\Omega((\log M)^{1 - \varepsilon})}$ for \[\varepsilon = \frac{(\log\log\log M)^{O(1)}}{\sqrt{\log\log M}} = o(1).\]
\end{theorem}

Since DTIME$\left(\exp{n^{O(1/\log\log n)}}\right) \subseteq$ $\cap_{\delta > 0}$DTIME$\left(\exp{n^\delta}\right)$, this gives hardness of approximation for the shortest vector problem in every finite $\ell_p$ norm assuming NP $\not\subseteq$ $\cap_{\delta > 0}$DTIME$\left(\exp{n^\delta}\right)$.

\paragraph{An Intermediate Lattice Problem.} We establish Theorem \ref{thm:main} by first showing constant factor hardness of approximation for a specific type of shortest vector problem:

\begin{theorem} \label{thm:constantgap}
    There is a deterministic $\exp{2^{O(\sqrt{\log n}(\log\log n)^{O(1)})}}$ time reduction from SAT instances of size $n$ to the following problem, where $M = \exp{2^{O(\sqrt{\log n}(\log\log n)^{O(1)})}}$ and $N = \exp{2^{O(\sqrt{\log n}(\log\log n)^{O(1)})}}$. Given a matrix $\mat B \in \mathbb{Z}^{M \times N}$ and a threshold $h$, distinguish between the following two cases:
    \begin{enumerate}
        \item There exists a nonzero vector $\mat x \in \mathbb{Z}^M$ such that $\|\mat x\mat B\|_0 \leq h$, \emph{and additionally} $\mat x \mat B \in \{-1,0,1\}^N$.
        \item For all nonzero vectors $\mat x \in \mathbb{Z}^M$, we have $\|\mat x \mat B\|_0 \geq 2h$.
    \end{enumerate}
\end{theorem}

In particular, case (1) enforces that the lattice vector is not only of small norm but also in $\{-1,0,1\}^N$. Theorem \ref{thm:constantgap} implies Theorem \ref{thm:main} via direct tensoring.

\begin{proof}[Proof of Theorem \ref{thm:main} using Theorem \ref{thm:constantgap}.]
    Fix any constant $p \geq 1$, and use the reduction from Theorem \ref{thm:constantgap} to get an integer matrix $\mat B \in \mathbb{Z}^{M \times N}$, where $1 \leq M \leq \exp{2^{O(\sqrt{\log n}(\log\log n)^{O(1)})}}$ and $1 \leq N \leq \exp{2^{O(\sqrt{\log n}(\log\log n)^{O(1)})}}$. Now let $\mat B' = \mat B^{\otimes n^{1/\log\log n}}$, i.e., the $n^{1/\log\log n}$-fold tensor product of $\mat B$ with itself.\footnote{We assume without loss of generality that $n^{1/\log\log n}$ is an integer.} This is the matrix of basis vectors for our final SVP instance.
    
    Notice that $\mat B' \in \mathbb{Z}^{M' \times N'}$, where
    \[M' = \exp{n^{1/\log\log n}\log M} \leq \exp{n^{O(1/\log\log n)}}\]
    and
    \[N' = \exp{n^{1/\log\log n}\log N} \leq \exp{n^{O(1/\log\log n)}}.\]
    So the rank of the lattice generated by $\mat B'$, and the time to construct $\mat B'$, both satisfy the bound in the theorem statement.
    
    The number of nonzero entries in the sparsest vector of a lattice is multiplicative under tensoring, so cases (1) and (2) from Theorem \ref{thm:constantgap} become:
    \begin{enumerate}
        \item There exists a nonzero vector $\mat x' \in \mathbb{Z}^{M'}$ such that $\|\mat x'\mat B'\|_0 \leq h^{n^{1/\log\log n}}$. \label{case:new1}
        \item For all nonzero vectors $\mat x' \in \mathbb{Z}^{M'}$, we have $\|\mat x'\mat B'\|_0 \geq 2^{n^{1/\log\log n}}h^{n^{1/\log\log n}}$. \label{case:new2}
    \end{enumerate}
    The first case also maintains that $\mat x'\mat B' \in \{-1,0,1\}^{N'}$.
    
    For all finite $p \geq 1$, the $\ell_p$ norm of an integer vector $\mat w'$ having $x$ nonzero entries is at least $x^{1/p}$, and this is realized with equality when the entries of $\mat w'$ are in $\{-1,0,1\}$. So cases (1) and (2) from above have the following dichotomy, where we define
    \[h' = h^{n^{1/\log\log n}/p},\]
    and change the parameters slightly to enforce strict inequality for case (2):\footnote{We assume without loss of generality that $h \geq 2$ and hence $h' \geq 2$; otherwise the problem is trivial.}
    \begin{enumerate}
        \item There exists a nonzero vector $\mat x' \in \mathbb{Z}^{M'}$ such that $\|\mat x'\mat B'\|_p \leq h'$.
        \item For all nonzero vectors $\mat x' \in \mathbb{Z}^{M'}$, we have $\|\mat x'\mat B'\|_p > 2^{n^{1/\log\log n}/p - 1}h'$.
    \end{enumerate}
    All that remains is to lower bound the multiplicative gap $\gamma = 2^{n^{1/\log\log n}/p - 1}$ in terms of the rank $M'$.

    \begin{claim} \label{claim:calculategamma}
        Let $p \geq 1$ be any constant, let $M$ be a parameter such that $1 \leq M \leq \exp{2^{O(\sqrt{\log n}(\log\log n)^{O(1)})}}$, and let $M' = \exp{n^{1/\log\log n}\log M}$. Then
        \[2^{n^{1/\log\log n}/p - 1} \geq \exp{\Omega\left((\log M')^{1 - (\log\log\log M')^{O(1)}/\sqrt{\log\log M'}}\right)}.\]
    \end{claim}
    We defer the calculations to Appendix \ref{app:proof:calculategamma}.
\end{proof}

\section{Vandermonde Fortified Tensor Product} \label{sec:VFtensor}

Our main technical contribution is a new type of tensor product, dubbed the \emph{Vandermonde fortified tensor product}, which interacts synergistically with a specific type of densest sub-hypergraph problem. This problem is just a different way to phrase the ``inner verifier'' from Khot's quasi-random PCP \cite{khot2006ruling}. It will be the starting point for our proof of Theorem \ref{thm:constantgap}.

\begin{problemnormal}[Quasi-Random Densest Sub-Hypergraph (QRDH), implicit in~\cite{khot2006ruling}] \label{prob:qrdh}
    Given a hypergraph $H = (V, E)$ of arity $d$ with $\vert V\vert = N$ and $\vert E\vert = M$, along with a parameter $r \geq 1$ and parameters $\alpha, \beta \in [0, 1]$, distinguish between the following two cases:
    \begin{enumerate}
        \item There exists a vertex subset $V' \subseteq V$ of size at most $N/r$ that fully contains at least $\alpha(1/r)^{d-1} M$ hyperedges.
        \item Every vertex subset $V' \subseteq V$ fully contains at most $(\vert V'\vert/N)^d M + \beta M$ hyperedges.
    \end{enumerate}
    As shorthand, we refer to the problem as $(N, M, d, r, \alpha, \beta)$-QRDH. We refer to $r$ and $\alpha$ as the \emph{completeness parameters}, and $\beta$ as the \emph{soundness parameter}.
\end{problemnormal}

\paragraph{Tensor Products and Indicator Matrices.} We are interested in computing tensor products of hypergraph indicator matrices.

\begin{definition}[Hypergraph Indicator Matrix]
    Given a hypergraph $H$ of arity $d$ with $N$ vertices and $M$ hyperedges, its indicator matrix is a matrix $\mat P \in \{0,1\}^{M \times N}$. There is one column for every vertex, one row for every hyperedge, and each row has $d$ nonzero entries to indicate the endpoints of the hyperedge.
\end{definition}

For an index $e \in [M]$, we use the notation $N^{(1)}_{\mat P}(e) \subseteq [N]$ to denote the set of $d$ columns supporting the nonzero entries of row $e$ in $\mat P$. For a subset $E' \subseteq [M]$, we use $N^{(1)}_{\mat P}(E')$ to denote the union of $N^{(1)}_{\mat P}(e)$ for each $e \in E'$.

We denote the $q$-fold tensor product of $\mat P$ with itself as $\mat Q = \mat P^{\otimes q}$. $\mat Q$ has $N^q$ columns, which we index using members of $[N]^q$; these should be interpreted as $q$-tuples of columns from the original matrix $\mat P$. Similarly, $\mat Q$ has $M^q$ rows, which we index using members of $[M]^q$ and interpret as $q$-tuples of rows from the original matrix $\mat P$.

Given a row index $\mat e \in [M]^q$, we use $N^{(q)}_{\mat P}(\mat e)$ to denote the set of columns supporting the nonzero entries of row $\mat e$ in $\mat P^{\otimes q}$. More precisely, for each $\mat e = (e_1, \ldots, e_q) \in [M]^q$, we have that
\[N^{(q)}_{\mat P}(\mat e) = N^{(1)}_{\mat P}(e_1) \times \ldots \times N^{(1)}_{\mat P}(e_q) \subseteq [N]^q.\]
Note that $\vert N^{(q)}_{\mat P}(\mat e)\vert = d^q$ for all $\mat e \in [M]^q$. Given a subset $E' \subseteq [M]^q$, we use $N^{(q)}_{\mat P}(E')$ to denote the union of $N^{(q)}_{\mat P}(\mat e)$ for all $\mat e \in E'$.

We can view indices $\mat e = (e_1, \ldots, e_q) \in [M]^q$ as individual hyperedges in the hypergraph represented by $\mat Q$. We refer to each of $e_1, \ldots, e_q$ as the $1$\textsuperscript{st}, \ldots, $q$\textsuperscript{th} sub-edges of $\mat e$, respectively. Similarly, we can view an index $\mat v = (v_1, \ldots, v_q) \in [N]^q$ as an individual vertex in the hypergraph represented by $\mat Q$, and each of $v_1, \ldots, v_q$ is a sub-vertex. For $q \geq 1$, we say that an index (hyperedge) $\mat e \in [M]^q$ is \emph{supported} on a set of columns (set of vertices) $V' \subseteq [N]^q$ if $N_{\mat P}^{(q)}(\mat e) \subseteq V'$. An analogous definition holds for subsets of indices (hyperedges) being supported on subsets of columns (subsets of vertices).

\paragraph{Vandermonde Gadgets.} One of our building blocks is a specific type of Vandermonde matrix over the integers.

\begin{definition}[Reduced Vandermonde Matrix]
    Let $a, b$ be positive integers such that $a > b$ and $a$ is prime. We define an $(a, b)$ reduced Vandermonde matrix $\mat V \in \mathbb{Z}^{(a - 1) \times b}$ as $\mat V_{i, j} = i^{j-1} \mod a$.
\end{definition}

Observe that reduced Vandermonde matrices can be computed deterministically in time $a^{O(1)}$. We have the following useful property:
\begin{lemma} \label{lem:vandindependent}
    Every $b \times b$ submatrix of an $(a, b)$ reduced Vandermonde matrix is of full rank.
\end{lemma}

\begin{proof}
    Well-known; see for example \cite{horn1994topics}. The idea is that an $(a - 1) \times b$ Vandermonde matrix over $\mathbb{F}_a$ has every $b \times b$ submatrix being of full rank, and casting from $\mathbb{F}_a$ to $\mathbb{Z}$ cannot introduce new linear dependencies.
\end{proof}

\paragraph{Defining the VF Tensor Product.} With the above tools at hand, we are ready to define the Vandermonde fortified tensor product. See Figure \ref{fig:VFTP} for an example. 

\begin{definition}[$(t, q)$-VF Tensor Product] \label{defn:VFTP}
    Given a matrix $\mat P \in \{0,1\}^{M \times N}$, along with positive integer parameters $t$ and $q$ such that $t$ divides $M$, the $(t, q)$-VF tensor product of $\mat P$ is constructed as follows.
    \begin{enumerate}
        \item Let $\mat Q = \mat P^{\otimes q}$ be the $q$-fold tensor product of $\mat P$ with itself.
        \item Define a collection $\mathcal{S}$ of $qM^{q-1}$ subsets of $[M]^q$:
        \[\mathcal{S} = \Big\{\{(e_1, \ldots, e_{\ell-1}, e_\ell', e_{\ell+1}, \ldots, e_q) : e_\ell' \in [M]\} \quad : \quad \ell \in [q] \text{ and } e_1, \ldots, e_{\ell-1}, e_{\ell+1}, \ldots, e_q \in [M]\Big\}.\]
        That is, each set in $\mathcal{S}$ is constructed by choosing a designated coordinate $\ell$, fixing the values $e_1, \ldots, \\ e_{\ell-1}, e_{\ell+1}, \ldots, e_q$ of all other coordinates, and allowing the designated coordinate to vary over all choices in $[M]$.
        \item Let $a$ be a prime satisfying $M^q < a < 3M^q$,\footnote{We use the range $(M^q, 3M^q)$ and not $(M^q, 2M^q)$ because of the edge case where $M = 1, q = 1$.} and let $\mat V$ be an $(a, M/t)$ reduced Vandermonde matrix. Index the first $M^q$ rows of $\mat V$ using distinct elements of $[M]^q$. Now for each $S \in \mathcal{S}$, define an $M^q \times M/t$ integer matrix $\mat A^{(S)}$ as follows. For all $\mat e \in [M]^q$, if $\mat e \in S$ then set $\mat A^{(S)}_{\mat e} = \mat V_{\mat e}$, and otherwise set $\mat A^{(S)} = \mat 0^{M/t}$.
        \item Denote the horizontal concatenation of all $\mat A^{(S)}$ matrices as $\mat A \in \mathbb{Z}^{M^q \times qM^q/t}$. The $(t,q)$-VF tensor product of $\mat P$ is $\mat T = \big[\mat A \| \mat Q\big]$.
    \end{enumerate}
\end{definition}

The synergy between VF tensor products and instances of QRDH is formalized below. Roughly speaking, VF tensor products allow us to amplify the gap between cases (1) and (2) in the QRDH problem, without creating spurious vertex subsets in case (2) that contain too many hyperedges.

\begin{theorem} \label{thm:vandtensor}
    Let $H = (V, E)$ be a hypergraph of arity $d$ with $N$ vertices and $M$ hyperedges, let $\beta \in [0,1]$, and let $t$ and $q$ be positive integer parameters. Suppose that
    \begin{enumerate}
        \item $t$ divides $M$,
        \item $1/t > \beta$, and
        \item (Expansion) Every vertex subset $V' \subseteq V$ fully contains at most $(\vert V'\vert/N)^d M + \beta M$ hyperedges.
    \end{enumerate} Denote $H$'s indicator matrix as $\mat P \in \{0,1\}^{M \times N}$, and let $\mat T = \big[\mat A \|\mat Q\big]$ be the $(t, q)$-VF tensor product of $\mat P$.

    The following holds for all vectors $\mat x \in \mathbb{Z}^{M^q}$ such that $\mat x \mat A = \mat 0$, and all parameters $\delta \in [0,1]$. Let $E'$ be the subset of $[M]^q$ containing all indices for the nonzero entries of $\mat x$. If $\vert E'\vert > \frac{M^q\delta^d}{(1 - \beta t)^q}$, then the rows of $\mat Q$ indexed by $E'$ have their nonzero entries supported on more than $N^q\delta$ distinct columns.
\end{theorem}

The purpose of having matrix $\mat A$ in the Vandermonde fortified tensor product is to impose some structure on the possible sets $E'$. In particular, we would like to enforce the following.

\begin{definition}
    Let $E' \subseteq [M]^q$. We say that this subset is \emph{$(t,q)$-legal} if for all $\ell \in [q]$, for all $e_1, \ldots, e_{\ell-1}, \\e_{\ell+1}, \ldots, e_q \in [M]$, the intersection
    \[E' \cap \{(e_1, \ldots, e_{\ell-1}, e_\ell', e_{\ell+1}, \ldots, e_q) : e_\ell' \in [M]\}\]
    either has cardinality $0$ or cardinality at least $M/t$.
\end{definition}

The lemma below uses the linear independence properties of reduced Vandermonde matrices.

\begin{lemma} \label{lem:mustbelegal}
    Let $H = (V, E)$ be a hypergraph of arity $d$ with $N$ vertices and $M$ hyperedges, let $\beta \in [0,1]$, and let $t$ and $q$ be positive integer parameters. Suppose that
    \begin{enumerate}
        \item $t$ divides $M$, and
        \item $1/t > \beta$.\footnote{This lemma does not require an expansion condition.}
    \end{enumerate} Denote $H$'s indicator matrix as $\mat P \in \{0,1\}^{M \times N}$, and let $\mat T = \big[\mat A \|\mat Q\big]$ be the $(t, q)$-VF tensor product of $\mat P$.
    
    The following holds for all vectors $\mat x \in \mathbb{Z}^{M^q}$ such that $\mat x \mat A = \mat 0$. Let $E'$ be the subset of $[M]^q$ containing all indices for the nonzero entries of $\mat x$. Then $E'$ is $(t,q)$-legal.
\end{lemma}

\begin{proof}
    We prove the contrapositive, namely: if $\mat x$ is not $(t,q)$-legal, then $\mat x\mat A \neq \mat 0$. By definition of legality, if $\mat x$ is not $(t,q)$-legal then there must exist an index $\ell \in [q]$ and indices $e_1, \ldots, e_{\ell-1}, e_{\ell+1}, \ldots, e_q \in [M]$ such that the set
    \[S = \{(e_1, \ldots, e_{\ell-1}, e_\ell', e_{\ell+1}, \ldots, e_q) : e_\ell' \in [M]\}\]
    satisfies
    \[0 < \vert E' \cap S \vert < M/t.\]
    By construction of the $(t, q)$-VF tensor product, there exists a column-induced submatrix $\mat A^{(S)}$ of $\mat A$ where all rows are set to $\mat 0^{M/t}$, except for those indexed by $S$, which are instead set to distinct rows of a width-$M/t$ reduced Vandermonde matrix. By our choice of the set $S$, we know that the matrix-vector product $\mat x\mat A^{(S)}$ corresponds to a linear combination that assigns a nonzero coefficient to at least 1, and strictly less than $M/t$, of these reduced Vandermonde rows. By Lemma \ref{lem:vandindependent}, any such set of rows is linearly independent, meaning that $\mat x\mat A^{(S)}$ is nonzero and hence $\mat x\mat A$ is also nonzero.
\end{proof}

A final ingredient that will be useful in proving Theorem \ref{thm:vandtensor} is the following simple observation:

\begin{lemma} \label{lem:xx'y}
    Let $x' \geq x > 0$ and $y \geq 0$. Then $(x' + y)/x' \leq (x+y)/x$.
\end{lemma}

\subsection{Proving the VF Tensor Product Theorem} \label{sec:prove3.5}
We are now ready to prove Theorem \ref{thm:vandtensor}. By Lemma \ref{lem:mustbelegal}, it will be sufficient to prove the following (slightly stronger) theorem:

\begin{theorem}\label{thm:legalexpand}
    Let $H = (V, E)$ be a hypergraph of arity $d$ with $N$ vertices and $M$ hyperedges, let $\beta \in [0,1]$, and let $t$ be a positive integer parameter. Suppose that
    \begin{enumerate}
        \item $t$ divides $M$,
        \item $1/t > \beta$, and
        \item (Expansion) Every vertex subset $V' \subseteq V$ fully contains at most $(\vert V'\vert/N)^d M + \beta M$ hyperedges.
    \end{enumerate} Denote $H$'s indicator matrix as $\mat P \in \{0,1\}^{M \times N}$.
    
    For all positive integers $q$, for all $(t,q)$-legal subsets $E' \subseteq [M]^q$, and for all parameters $\delta \in [0,1]$, the following holds. If $\vert E'\vert > \frac{M^q\delta^d}{(1 - \beta t)^q}$, then $\vert N^{(q)}_{\mat P}(E') \vert > N^q\delta$.
\end{theorem}

\begin{proof}
    We give a proof by induction on the tensoring exponent $q$, working with the contrapositive of the claim in the theorem, namely: if $\vert N^{(q)}_{\mat P}(E') \vert \leq N^q\delta$, then $\vert E'\vert \leq \frac{M^q\delta^d}{(1 - \beta t)^q}$. Define a function $L^{(q)}_{\mat P}(x)$, which takes as input a number $x$ and outputs the\footnote{Technically, there may be multiple subsets $E'$ of maximum size.} cardinality of the largest $(t, q)$-legal subset $E' \subseteq [M]^q$ such that $\vert N^{(q)}_{\mat P}(E')\vert \leq x$. Our inductive hypothesis is that for all $\delta \in [0,1]$, we have $L^{(q)}_{\mat P}(N^q\delta) \leq \frac{M^q\delta^d}{(1 - \beta t)^q}$.

    \paragraph{Base case $q = 1$.} Let $\delta \in [0,1]$ be any parameter, and let $E' \subseteq [M]$ be any $(t, 1)$-legal subset such that $\vert N^{(1)}_{\mat P}(E')\vert \leq N\delta$. Our goal is to show that $\vert E'\vert \leq \frac{M\delta^d}{1-\beta t}$.

    Since $E'$ is $(t, 1)$-legal, we have that either $\vert E'\vert = 0$ or $\vert E'\vert \geq M/t$. In the first case, the claim is automatically satisfied for all choices of $\delta$ (our assumption that $1/t > \beta$ implies that $1 - \beta t > 0$, so the upper bound on $\vert E'\vert$ is always nonnegative). Now assume that $\vert E'\vert \geq M/t$.
    
    Using the expansion assumption $\mat P$, we have that
    \begin{align*}
        \vert E'\vert \leq & \left(\frac{\vert N^{(1)}_{\mat P}(E')\vert}{N}\right)^dM + \beta M \\
        \leq & \delta^dM + \beta M \tag{Assumed $\vert N^{(1)}_{\mat P}(E')\vert \leq N\delta$} \\
        \leq & \left(\frac{\delta^d + \beta}{\delta^d}\right) \delta^d M.
    \end{align*}

    Since $\vert E'\vert \geq M/t$, the above implies that $\delta^d \geq (1/t-\beta)$. We also know that $1/t - \beta$ is strictly positive by assumption. Using these along with Lemma \ref{lem:xx'y}, we can write:
    
    \begin{align*}
        \left(\frac{\delta^d + \beta}{\delta^d}\right) \delta^d M & \leq \left(\frac{(1/t - \beta) + \beta}{(1/t - \beta)}\right) \delta^d M \\
        & \leq \frac{\delta^d M}{1 - \beta t}.
    \end{align*}

    So when $\vert E'\vert \geq M/t$, we have $\vert E'\vert \leq \frac{M\delta^d}{1-\beta t}$, and the base case holds.

    \paragraph{Inductive Case $q > 1$: Setup.} 
    
    We start by defining a series of slice operators for subsets $V' \subseteq [N]^q$ and $E' \subseteq [M]^q$. Below, let $v_1, \ldots, v_q \in [N]$ and $e_1, \ldots, e_q \in [M]$ be any indices.
    \begin{align*}
        W^{(q)}_{v_1}(V') & \coloneqq \{(v_2', \ldots, v_q') \in [N]^{q-1} : (v_1, v_2', \ldots, v_q') \in V'\} \\
        X^{(q)}_{v_2, \ldots, v_q}(V') & \coloneqq \{v_1' \in [N] : (v_1', v_2, \ldots, v_q) \in V'\} \\
        Y^{(q)}_{e_1}(E') & \coloneqq \{(e_2', \ldots, e_q') \in [M]^{q-1} : (e_1, e_2', \ldots, e_q') \in E'\} \\
        Z^{(q)}_{e_2, \ldots, e_q}(E') & \coloneqq \{e_1' \in [M] : (e_1', e_2, \ldots, e_q) \in E'\}
    \end{align*}
    
    We argue that slices of legal subsets of $[M]^q$ are themselves legal, with respect to the appropriate tensoring parameter.

    \begin{claim} \label{claim:slicelegal}
        Let $E' \subseteq [M]^q$ be any $(t, q)$-legal subset. Then for all $e_1, \ldots, e_q \in [M]$, the slice $Y^{(q)}_{e_1}(E')$ is $(t, q-1)$-legal, and the slice $Z^{(q)}_{e_2, \ldots, e_q}(E')$ is $(t, 1)$-legal.
    \end{claim}

    \begin{proof}
        We prove that $Y^{(q)}_{e_1}(E')$ is $(t, q-1)$-legal, as the proof for $Z^{(q)}_{e_2, \ldots, e_q}(E')$ is nearly identical. Suppose for contradiction that $Y^{(q)}_{e_1}(E')$ was not $(t, q-1)$-legal. Then by definition of (il)legality, there exists $\ell \in [2, q]$ and a choice of $e_2^*, \ldots, e_{\ell-1}^*, e_{\ell+1}^*, \ldots, e_q^*$ such that
        \[0 < \vert Y^{(q)}_{e_1}(E') \cap \{(e_2^*, \ldots, e_{\ell-1}^*, e_\ell', e_{\ell+1}^*, \ldots, e_q^*) : e_\ell' \in [M]\} \vert < M/t.\]
        By construction of $Y^{(q)}_{e_1}(E')$, this implies that
        \[0 < \vert E' \cap \{(e_1, e_2^*, \ldots, e_{\ell-1}^*, e_\ell', e_{\ell+1}^*, \ldots, e_q^*) : e_\ell' \in [M]\} \vert < M/t,\]
        which violates the assumption that $E'$ is $(t, q)$-legal.
    \end{proof}
    
    Using Claim \ref{claim:slicelegal}, we have the following dichotomy for slices of $N^{(q)}_{\mat P}(E')$ when $E'$ is $(t, q)$-legal.

    \begin{claim} \label{claim:minnonzero}
        Let $E' \subseteq [M]^q$ be any $(t, q)$-legal subset, and set $V' = N^{(q)}_{\mat P}(E')$. Then for all $v_2, \ldots, v_q \in [N]$, we have that $\vert X^{(q)}_{v_2, \ldots, v_q}(V')\vert$ is either 0 or at least $(1/t - \beta)^{1/d} N$.
    \end{claim}

    \begin{proof}
        Let $v_2, \ldots, v_q \in [N]$ be any indices, which will stay fixed for the rest of the proof. If $\vert X^{(q)}_{v_2, \ldots, v_q}(V')\vert = 0$, then we are done, so assume that $\vert X^{(q)}_{v_2, \ldots, v_q}(V')\vert > 0$. Our goal will be to show that $\vert X^{(q)}_{v_2, \ldots, v_q}(V')\vert \geq (1/t-\beta)^{1/d}N$. By definition of $V'$ and the assumption that $X^{(q)}_{v_2, \ldots, v_q}(V')$ is nonempty, we know that there exists $\mat e = (e_1, e_2, \ldots, e_q) \in E'$ such that
        \[N^{(q)}_{\mat P}(\mat e) \cap \{(v_1', v_2, \ldots, v_q) : v_1' \in X^{(q)}_{v_2, \ldots, v_q}(V')\} \neq \emptyset.\]
        In other words, there exists at least one hyperedge $\mat e = (e_1, e_2, \ldots, e_q) \in E'$ incident to a vertex $(v_1', v_2, \ldots, v_q)$ such that $v_1' \in X^{(q)}_{v_2, \ldots, v_q}(V')$. Fix this choice of $e_2, \ldots, e_q$ for the rest of the proof.

        Now consider the slice $Z^{(q)}_{e_2, \ldots, e_q}(E')$. Recall that, from the hypergraph perspective, this is constructed by first taking the set of all hyperedges in $E'$ whose last $q-1$ sub-edges are $e_2, \ldots, e_q$, and then deleting all but the first sub-edge from each hyperedge. We know that $Z^{(q)}_{e_2, \ldots, e_q}(E')$ is $(t, 1)$-legal by Claim \ref{claim:slicelegal}, and by construction it is nonempty. Therefore
        \[\vert Z^{(q)}_{e_2, \ldots, e_q}(E')\vert \geq M/t.\]
        By definition of $V'$, each sub-edge $e_1' \in Z^{(q)}_{e_2, \ldots, e_q}(E')$ has all of its endpoints $N_{\mat P}^{(1)}(e_1')$ appearing in $X^{(q)}_{v_2, \ldots, v_q}(V')$, meaning that
        \[N^{(1)}_{\mat P}(Z^{(q)}_{e_2, \ldots, e_q}(E')) \subseteq X^{(q)}_{v_2, \ldots, v_q}(V').\]
        So to prove the claim, it will be sufficient to argue that every subset $E'' \subseteq [M]$ of size at least $M/t$ satisfies $\vert N^{(1)}_{\mat P}(E'')\vert \geq (1/t-\beta)^{1/d}N$. By the expansion assumption on $\mat P$, we know that
        \[\vert E''\vert \leq \left(\frac{\vert N^{(1)}_{\mat P}(E'')\vert}{N}\right)^dM + \beta M.\]
        Substituting $\vert E''\vert \geq M/t$ and re-arranging gives
        \begin{align*}
            M/t \leq & \left(\frac{\vert N^{(1)}_{\mat P}(E'')\vert}{N}\right)^dM + \beta M \\
            (1/t - \beta)M \leq & \left(\frac{\vert N^{(1)}_{\mat P}(E'')\vert}{N}\right)^dM \\
            (1/t - \beta)^{1/d} \leq & \vert N^{(1)}_{\mat P}(E'')\vert/N \\
            (1/t - \beta)^{1/d}N \leq & \vert N^{(1)}_{\mat P}(E'')\vert.
        \end{align*}
        This is well-defined since we assumed that $1/t-\beta > 0$.
    \end{proof}
    
    \paragraph{Inductive Case $q > 1$: Bounding $L^{(q)}_{\mat P}(N^q\delta)$.} Recall that we defined a function $L^{(q)}_{\mat P}(x)$, which takes as input a number $x$ and outputs the cardinality of the largest $(t, q)$-legal subset $E' \subseteq [M]^q$ such that $\vert N^{(q)}_{\mat P}(E')\vert \leq x$. By induction, we know that for all $\delta \in [0,1]$ and all $q' < q$, we have $L^{(q')}_{\mat P}(N^{q'}\delta) \leq \frac{M^{q'}\delta^d}{(1 - \beta t)^{q'}}$. Our goal is to prove that for all $\delta \in [0,1]$, we have $L^{(q)}_{\mat P}(N^q\delta) \leq \frac{M^q\delta^d}{(1-\beta t)^q}$. To this end, let $E' \subseteq [M]^q$ be any $(t, q)$-legal subset such that $\vert N^{(q)}_{\mat P}(E')\vert \leq N^q\delta$, and set $V' = N^{(q)}_{\mat P}(E')$. For notational purposes, define a parameter
    \[\Gamma = \frac{M^{q-1}}{N^{d(q-1)}(1-\beta t)^{q-1}}.\]

    We first upper bound $\vert E'\vert$ in terms of $(t, q-1)$-legal subsets supported on slices of $V'$, allowing us to invoke the inductive hypothesis with $q' = q-1$. 

    \begin{claim} \label{claim:firstslices}
        \[\vert E'\vert \leq \Gamma \cdot \sum_{e_1 \in [M]} \vert \cap_{v_1 \in N^{(1)}_{\mat P}(e_1)} W^{(q)}_{v_1}(V')\vert^d.\]
    \end{claim}

    \begin{proof}
    By Claim \ref{claim:slicelegal}, we know that for any choice of $e_1 \in [M]$, the slice $Y^{(q)}_{e_1}(E')$ is $(t, q-1)$-legal. Notice that we can write $\vert N^{(q-1)}_{\mat P}(Y^{(q)}_{e_1}(E'))\vert = N^{q-1} \delta$, where
    \[\delta =\frac{\vert N^{(q-1)}_{\mat P}(Y^{(q)}_{e_1}(E'))\vert}{N^{q-1}}. \]
    So by the inductive hypothesis,
    \[\vert Y^{(q)}_{e_1}(E') \vert \leq \left(\frac{\vert N^{(q-1)}_{\mat P}(Y^{(q)}_{e_1}(E'))\vert}{N^{q-1}}\right)^d \cdot M^{q-1}/(1-\beta t)^{q-1}.\]
    Ranging over all choices of $e_1 \in [M]$, the slices $Y^{(q)}_{e_1}(E')$ partition $E'$, so we have
    \begin{align} \label{eq:1}
        \vert E'\vert \leq \sum_{e_1 \in [M]} \left(\frac{\vert N^{(q-1)}_{\mat P}(Y^{(q)}_{e_1}(E'))\vert}{N^{q-1}}\right)^d \cdot M^{q-1}/(1-\beta t)^{q-1}.
    \end{align}

    By definition of $V'$, for all $e_1 \in [M]$ and $v_1 \in N^{(1)}_{\mat P}(e_1)$, we have
    \[N^{(q-1)}_{\mat P}(Y^{(q)}_{e_1}(E')) \subseteq W^{(q)}_{v_1}(V'),\]
    which implies that
    \[N^{(q-1)}_{\mat P}(Y^{(q)}_{e_1}(E')) \subseteq \cap_{v_1 \in N^{(1)}_{\mat P}(e_1)} W^{(q)}_{v_1}(V').\]
    Thus we can re-write Equation (\ref{eq:1}) as
    \begin{align*}
        \vert E'\vert \leq \sum_{e_1 \in [M]} \left(\frac{\vert \cap_{v_1 \in N^{(1)}_{\mat P}(e_1)} W^{(q)}_{v_1}(V')\vert}{N^{q-1}}\right)^d \cdot M^{q-1}/(1-\beta t)^{q-1}.
    \end{align*}
    Using that
    \[\Gamma = \frac{M^{q-1}}{N^{d(q-1)}(1-\beta t)^{q-1}},\]
    this becomes
    \begin{align*}
        \vert E'\vert \leq \Gamma \cdot \sum_{e_1 \in [M]} \vert \cap_{v_1 \in N^{(1)}_{\mat P}(e_1)} W^{(q)}_{v_1}(V')\vert^d.
    \end{align*}
    \end{proof}

    We now convert from a summation in terms of $W^{(q)}_{v_1}(V')$ slices to a summation in terms of $X^{(q)}_{v_2, \ldots, v_q}(V')$ slices. Notice that for any $e_1 \in [M]$, we have 
    \[\vert \cap_{v_1 \in N^{(1)}_{\mat P}(e_1)} W^{(q)}_{v_1}(V')\vert^d = \sum_{\mat v^{(1)}, \ldots, \mat v^{(d)} \in [N]^{q-1}} \left(\prod_{i \in [d]} \mat{1}_{N^{(1)}_{\mat P}(e_1) \subseteq X^{(q)}_{\mat v^{(i)}}(V')}\right).\]
    In other words, the following two processes are equivalent:
    \begin{enumerate}
        \item \emph{For each endpoint $v_1$ of the hyperedge $e_1$, take the set of all vertices in $V'$ whose first sub-vertex is $v_1$. Then intersect the resulting sets, and raise the cardinality of their intersection to the power of $d$.}
        \item \emph{Say that a $(q-1)$-tuple $(v_2, \ldots, v_q) \in [N]^{q-1}$ supports a hyperedge $e_1 \in [M]$ whenever every endpoint $v_1$ of $e_1$ satisfies that $(v_1, v_2, \ldots, v_q) \in V'$. Now count the number of $d$-tuples of $(q-1)$-tuples $((v_2, \ldots, v_q)^{(1)}, \ldots, (v_2, \ldots, v_q)^{(d)})$ such that every constituent $(q-1)$-tuple $(v_2, \ldots, v_q)^{(i)}$ supports $e_1$.}
    \end{enumerate}
    Plugging this into the inequality from Claim \ref{claim:firstslices}, we have
    \begin{align*}
        \vert E'\vert \leq \Gamma \cdot \sum_{e_1 \in [M]} \left(\sum_{\mat v^{(1)}, \ldots, \mat v^{(d)} \in [N]^{q-1}} \left(\prod_{i \in [d]} \mat{1}_{N^{(1)}_{\mat P}(e_1) \subseteq X^{(q)}_{\mat v^{(i)}}(V')}\right)\right).
    \end{align*}
    Switching the order of summation gives:
    \begin{align} \label{eq:2}
        \vert E'\vert \leq \Gamma \cdot \sum_{\mat v^{(1)}, \ldots, \mat v^{(d)} \in [N]^{q-1}} \left(\sum_{e_1 \in [M]} \left(\prod_{i \in [d]} \mat{1}_{N^{(1)}_{\mat P}(e_1) \subseteq X^{(q)}_{\mat v^{(i)}}(V')}\right)\right).
    \end{align}

    Below we give a more useful upper bound on the inner summation:

    \begin{claim} \label{claim:usefulupper}
        Let $\mat v^{(1)}, \ldots, \mat v^{(d)} \in [N]^{q-1}$. Then
        \[\sum_{e_1 \in [M]} \left(\prod_{i \in [d]} \mat{1}_{N^{(1)}_{\mat P}(e_1) \subseteq X^{(q)}_{\mat v^{(i)}}(V')}\right) \leq \left(\prod_{i \in [d]}\vert X^{(q)}_{\mat v^{(i)}}(V')\vert\right)\frac{M}{N^d(1 - \beta t)}.\]
    \end{claim}

    \begin{proof}
    Notice that the left hand side simply counts the number of indices $e_1 \in [M]$ such that $N^{(1)}_{\mat P}(e_1)$ is contained within the set $\cap_{i \in [d]} X^{(q)}_{\mat v^{(i)}}(V').$
    If $\cap_{i \in [d]} X^{(q)}_{\mat v^{(i)}}(V')$ is empty, then
    \[\sum_{e_1 \in [M]} \left(\prod_{i \in [d]} \mat{1}_{N^{(1)}_{\mat P}(e_1) \subseteq X^{(q)}_{\mat v^{(i)}}(V')}\right) = 0,\]
    in which case the claim holds automatically because the upper bound is always non-negative. So assume that $\cap_{i \in [d]} X^{(q)}_{\mat v^{(i)}}(V')$ is nonempty. Here, we can use the expansion assumption on $\mat P$ to get the upper bound
    \begin{align*}
        \sum_{e_1 \in [M]} \left(\prod_{i \in [d]} \mat{1}_{N^{(1)}_{\mat P}(e_1) \subseteq X^{(q)}_{\mat v^{(i)}}(V')}\right) & \leq \left(\frac{\vert \cap_{i \in [d]} X^{(q)}_{\mat v^{(i)}}(V')\vert}{N}\right)^d M + \beta M \\
        & \leq \left(\frac{\prod_{i \in [d]}\vert X^{(q)}_{\mat v^{(i)}}(V')\vert^{1/d}}{N}\right)^d M + \beta M \tag{Upper bounding intersection size by geometric mean of sizes.}\\
        & \leq \left(\prod_{i \in [d]}\vert X^{(q)}_{\mat v^{(i)}}(V')\vert\right) \cdot M/N^d + \beta M.
    \end{align*}
    
    Since $\cap_{i \in [d]} X^{(q)}_{\mat v^{(i)}}(V')$ is nonempty, each $X^{(q)}_{\mat v^{(i)}}(V')$ is nonempty. Hence by Claim \ref{claim:minnonzero}, each of these sets must be of size at least $(1/t-\beta)^{1/d}N$. This implies that
    \[\prod_{i \in [d]}\vert X^{(q)}_{\mat v^{(i)}}(V')\vert \geq (1/t - \beta)N^d.\]
    Writing $\Lambda = \prod_{i \in [d]}\vert X^{(q)}_{\mat v^{(i)}}(V')\vert$ and using Lemma \ref{lem:xx'y}, we have
    \begin{align*}
        \left(\prod_{i \in [d]}\vert X^{(q)}_{\mat v^{(i)}}(V')\vert\right) \cdot M/N^d + \beta M & \leq \Lambda \cdot M/N^d + \beta M \\
        & \leq \left(\frac{\Lambda/N^d + \beta}{\Lambda/N^d}\right) \Lambda M/N^d \\
        & \leq \left(\frac{1/t}{1/t - \beta}\right) \Lambda M/N^d \\
        & \leq \Lambda \frac{M}{N^d(1 - \beta t)}. \\
    \end{align*}
    Putting everything together, we have that whenever $\cap_{i \in [d]} X^{(q)}_{\mat v^{(i)}}(V')$ is nonempty, then
    \[\sum_{e_1 \in [M]} \left(\prod_{i \in [d]} \mat{1}_{N^{(1)}_{\mat P}(e_1) \subseteq X^{(q)}_{\mat v^{(i)}}(V')}\right) \leq \left(\prod_{i \in [d]}\vert X^{(q)}_{\mat v^{(i)}}(V')\vert\right)\frac{M}{N^d(1 - \beta t)}.\]
    \end{proof}

    From here, the inductive upper bound follows by just combining Equation (\ref{eq:2}) with Claim \ref{claim:usefulupper} and then performing algebraic manipulations.
    \begin{align*}
        \vert E'\vert \leq & \Gamma \cdot \sum_{\mat v^{(1)}, \ldots, \mat v^{(d)} \in [N]^{q-1}} \left(\sum_{e_1 \in [M]} \left(\prod_{i \in [d]} \mat{1}_{N^{(1)}_{\mat P}(e_1) \subseteq X^{(q)}_{\mat v^{(i)}}(V')}\right)\right) \tag{Equation (\ref{eq:2}).} \\
        \leq & \Gamma \cdot \sum_{\mat v^{(1)}, \ldots, \mat v^{(d)} \in [N]^{q-1}} \left(\left(\prod_{i \in [d]}\vert X^{(q)}_{\mat v^{(i)}}(V')\vert\right)\frac{M}{N^d(1 - \beta t)}\right) \tag{Claim \ref{claim:usefulupper}.} \\
        \leq & \Gamma \cdot \frac{M}{N^d(1 - \beta t)} \cdot \sum_{\mat v^{(1)}, \ldots, \mat v^{(d)} \in [N]^{q-1}} \left(\prod_{i \in [d]}\vert X^{(q)}_{\mat v^{(i)}}(V')\vert\right) \\
        \leq & \frac{M^q}{N^{dq}(1 - \beta t)^q} \cdot \sum_{\mat v^{(1)}, \ldots, \mat v^{(d)} \in [N]^{q-1}} \left(\prod_{i \in [d]}\vert X^{(q)}_{\mat v^{(i)}}(V')\vert\right) \tag{$\Gamma = \frac{M^{q-1}}{N^{d(q-1)}(1-\beta t)^{q-1}}$} \\
        \leq & \frac{M^q}{N^{dq}(1 - \beta t)^q} \cdot \left(\sum_{\mat v \in [N]^{q-1}} \vert X^{(q)}_{\mat v}(V')\vert\right)^d \\
    \end{align*}
    Because the sets $X^{(q)}_{\mat v}(V')$ partition $V'$, we can write:
    \begin{align*}
        \vert E'\vert \leq & \frac{M^q}{N^{dq}(1 - \beta t)^q} \cdot \vert V'\vert^d \\
        \leq & \frac{M^q}{N^{dq}(1 - \beta t)^q} \cdot (N^q\delta)^d \tag{Set $V' = N^{(q)}_{\mat P}(E')$ and assumed $\vert N^{(q)}_{\mat P}(E')\vert \leq N^q\delta$}\\
        \leq & \frac{M^q\delta^d}{(1 - \beta t)^q} \\
    \end{align*}
    Thus the inductive upper bound holds.
\end{proof}

\section{Getting a Constant Gap for SVP} \label{sec:constgap}

In this section, we prove Theorem \ref{thm:constantgap}, which is re-stated below.

\paragraph{Theorem \ref{thm:constantgap}, Restated.} 
    \emph{There is a deterministic $\exp{2^{O(\sqrt{\log n}(\log\log n)^{O(1)})}}$ time reduction from SAT instances of size $n$ to the following problem, where $M = \exp{2^{O(\sqrt{\log n}(\log\log n)^{O(1)})}}$ and \\ $N = \exp{2^{O(\sqrt{\log n}(\log\log n)^{O(1)})}}$. Given a matrix $\mat B \in \mathbb{Z}^{M \times N}$ and a threshold $h$, distinguish between the following two cases:
    \begin{enumerate}
        \item There exists a nonzero vector $\mat x \in \mathbb{Z}^M$ such that $\|\mat x\mat B\|_0 \leq h$, \emph{and additionally} $\mat x \mat B \in \{-1,0,1\}^N$.
        \item For all nonzero vectors $\mat x \in \mathbb{Z}^M$, we have $\|\mat x \mat B\|_0 \geq 2h$.
    \end{enumerate}}
\noindent
For reference, we also re-state the QRDH problem.

\paragraph{Problem \ref{prob:qrdh}, Restated.}
    \emph{Given a hypergraph $H = (V, E)$ of arity $d$ with $\vert V\vert = N$ and $\vert E\vert = M$, along with a parameter $r \geq 1$ and parameters $\alpha, \beta \in [0, 1]$, distinguish between the following two cases:
    \begin{enumerate}
        \item There exists a vertex subset $V' \subseteq V$ of size at most $N/r$ that fully contains at least $\alpha(1/r)^{d-1} M$ hyperedges.
        \item Every vertex subset $V' \subseteq V$ fully contains at most $(\vert V'\vert/N)^d M + \beta M$ hyperedges.
    \end{enumerate}
    As shorthand, we refer to the problem as $(N, M, d, r, \alpha, \beta)$-QRDH. We refer to $r$ and $\alpha$ as the \emph{completeness parameters}, and $\beta$ as the \emph{soundness parameter}.}\\

We show in Section \ref{sec:qrdhproof} that a modification of Khot's quasi-random PCP \cite{khot2006ruling}, combined with tools from a paper by Khot and Saket \cite{khot2016hardness}, gives the following theorem.

\begin{theorem} \label{thm:khotqrdh}
    There is a deterministic $\exp{2^{O(\sqrt{\log n}\log^4\log n)}}$ time reduction from SAT instances of size $n$ to 
    \begin{align*}
        \big( & N = \exp{2^{O(\sqrt{\log n}\log^4\log n)}}, \\
        & M = \exp{2^{O(\sqrt{\log n}\log^4\log n)}}, \\
        & d = O(\sqrt{\log n}/\log^2\log n), \\
        & r = 2^{\Theta(\sqrt{\log n})}, \alpha = 0.5, \beta = n^{-\Omega(\log\log n)}\big)\textnormal{-QRDH}.
    \end{align*}
\end{theorem}

\paragraph{Setting Parameters and Starting the Reduction.}
Start with a SAT instance of size $n$, and apply the reduction in Theorem \ref{thm:khotqrdh} to obtain a hypergraph $H^-$ with $N^-$ vertices and $M^-$ hyperedges, which we represent using an indicator matrix $\mat {P^-} \in \{0,1\}^{M^- \times N^-}$. The hypergraph is of arity $d$, the completeness parameters are $r$ and $\alpha$, and the soundness parameter is $\beta$.

We now define a sequence of additional parameters. First fix a positive integer $t$ satisfying
\[ n/(\alpha(1/r)^{d-1}) \leq t \leq 2n/(\alpha(1/r)^{d-1}).\]
By the setting of parameters in Theorem \ref{thm:khotqrdh}, the lower bound on $t$ is superconstant, so an appropriate choice always exists assuming $n$ is sufficiently large.

The next step is to (slightly) modify the hypergraph $H^-$, so that the number of rows in the resulting hypergraph is divisible by $t$ (in addition to satisfying other useful properties). To this end, duplicate every hyperedge of $H^-$ exactly $n^2tN^-$ times to get a new hypergraph $H$, and leave the vertex set unchanged. $H$ has $N = N^-$ vertices and $M = Nn^2tM^-$ hyperedges, and we represent $H$ using an indicator matrix $\mat P \in \{0,1\}^{M \times N}$. Notice that $M$ is bounded as $\exp{2^{O(\sqrt{\log n}\log^4\log n)}}$, and this conversion maintains cases (1) and (2) of the QRDH problem. To be more precise:
\begin{enumerate}
    \item \emph{If there exists a vertex subset $V' \subseteq V$ of size at most $N/r$ that fully contains at least $\alpha(1/r)^{d-1}M^-$ hyperedges of $H^-$, then that same subset fully contains at least $\alpha(1/r)^{d-1}M$ hyperedges of $H$.}
    \item \emph{If every vertex subset $V' \subseteq V$ fully contains at most $(\frac{\vert V'\vert}{N})^d M^- + \beta M^-$ hyperedges of $H^-$, then every vertex subset $V' \subseteq V$ fully contains at most $(\frac{\vert V'\vert}{N})^d M + \beta M$ hyperedges of $H$.}
\end{enumerate}

Fix three additional positive integers $q, h,$ and $w$ satisfying:
\begin{align*}
    \log n/\log\log n \quad \leq & \quad q \quad \leq \quad 2\log n/\log\log n \\
    \lceil\alpha(1/r)^{d-1}M\rceil^q \quad \leq & \quad h \quad \leq \quad (2 \alpha(1/r)^{d-1}M)^q \\
    \left(\frac{\alpha(1/r)^{d-1}M}{N/r}\right)^q/n \quad \leq & \quad w \quad \leq \quad 2\left(\frac{\alpha(1/r)^{d-1}M}{N/r}\right)^q/n
\end{align*}
We also require that $n$ divides $h$. Notice that, by our choice of parameters and the fact that we duplicated each hyperedge of the original hypergraph $H^-$ exactly $Nn^2t$ times, we have
\begin{align*}
    \log n/\log\log n & \geq \omega(1)
\end{align*}
\begin{align*}
    \alpha(1/r)^{d-1}M & \geq \Theta\left(1/2^{\sqrt{\log n}}\right)^{O(\sqrt{\log n}/\log^2\log n)} \cdot Nn^2tM^- \\
    & \geq n^{2 - O(1/\log^2\log n)} \\
    & \geq \omega(n)
\end{align*}
\begin{align*}
    \left(\frac{\alpha(1/r)^{d-1}M}{N/r}\right)^q/n & \geq \left(\frac{\Theta\left(1/2^{\sqrt{\log n}}\right)^{O(\sqrt{\log n}/\log^2\log n)} \cdot Nn^2tM^-}{N/r}\right)^q/n \\
    & \geq n^{2 - O(1/\log^2\log n)}/n \\
    & \geq \omega(1)
\end{align*}
So an appropriate choice always exists, assuming that $n$ is sufficiently large.

\paragraph{Constructing the Lattice.} We construct the matrix $\mat B$ of basis vectors for our SVP instance as follows. The matrix will have $M^q$ rows, which (as in Section \ref{sec:VFtensor}) we index using members of $[M]^q$. See Figures \ref{fig:PQR} and \ref{fig:matC} for an example construction.
\begin{enumerate}
    \item Let $\mat T = \big[\mat A \|\mat Q\big]$ be the $(t, q)$-VF tensor product of $\mat P$. Notice that $\mat A$ is an integer matrix with $M^q$ rows and $qM^q/t$ columns, and $\mat Q$ is an integer matrix with $M^q$ rows and $N^q$ columns. All entries of $\mat A$ are nonnegative and less than $3M^q$, and $\mat Q$ is a 0/1 matrix.
    \item 
    Let $\mat V \in \mathbb{Z}^{M^q \times w}$ be an $(a, w)$ reduced Vandermonde matrix, where $a$ is a prime satisfying $M^q < a < 3M^q$. Index the first $M^q$ rows of $\mat V$ using distinct elements of $[M]^q$. For each $\mat v \in [N]^q$, we construct a matrix $\mat R^{(\mat v)} \in \mathbb{Z}^{M^q \times w}$ using $\mat Q$ as follows:
    \begin{itemize}
        \item For each $\mat e \in [M]^q$, do:
        \begin{itemize}
            \item If $\mat Q_{\mat e, \mat v} = 0$, then set $\mat R^{(\mat v)}_{\mat e} = \mat 0^w$.
            \item Otherwise, set $\mat R^{(\mat v)}_{\mat e} = \mat V_{\mat e}$.
        \end{itemize}
    \end{itemize}
    Now let $\mat R$ be the horizontal concatenation of the matrices $\{\mat R^{(\mat v)} : \mat v \in [N]^q\}$.
    
    In other words, to construct the matrix $\mat R$, we replace every $0$ entry in $\mat Q$ with the row vector $\mat 0^w$, and we replace every $1$ entry in $\mat Q$ with a row vector taken from a width-$w$ reduced Vandermonde matrix.
    \item Let $\mat W \in \mathbb{Z}^{M^q \times h/n}$ be an $(a, h/n)$ reduced Vandermonde matrix, where $a$ is a prime satisfying $M^q < a < 3M^q$.
    \item Let $\mat C = \big[\mat A \| \mat R \| \mat W\big] \in \mathbb{Z}^{M^q \times (qM^q/t + N^qw + h/n)}$.
    \item Construct the basis matrix $\mat B$ by horizontally concatenating $2h$ copies of $\mat C$, and then horizontally appending a single $M^q \times M^q$ identity matrix. For convenience, denote the number of rows in $\mat B$ as $M'$ and the number of columns in $\mat B$ as $N'$.
\end{enumerate}

\paragraph{Relating Back to Theorem \ref{thm:constantgap}.} By our choice of parameters, the time to construct $\mat B$ is $\exp{2^{O(\sqrt{\log n}(\log\log n)^{O(1)})}}$, and all steps (including the VF tensor product) are deterministic. We also have that $M' = \exp{2^{O(\sqrt{\log n}(\log\log n)^{O(1)})}}$ and $N' = \exp{2^{O(\sqrt{\log n}(\log\log n)^{O(1)})}}$. All that remains is to prove:

\begin{lemma}[Completeness] \label{lem:completeness}
    Suppose that the starting SAT instance was satisfiable, so that the hypergraph $H$ satisfies case (1) of the QRDH problem. Then assuming $n$ is sufficiently large, there exists a nonzero vector $\mat x \in \mathbb{Z}^{M'}$ such that $\|\mat x\mat B\|_0 \leq h$, and furthermore $\mat x \mat B \in \{-1,0,1\}^{N'}$.
\end{lemma}

\begin{lemma}[Soundness] \label{lem:soundness}
    Suppose that the starting SAT instance was unsatisfiable, so that the hypergraph $H$ satisfies case (2) of the QRDH problem. Then assuming $n$ is sufficiently large, for all nonzero $\mat x \in \mathbb{Z}^{M'}$, we have $\|\mat x\mat B\|_0 \geq 2h$.
\end{lemma}

\subsection{Proof of Lemma \ref{lem:completeness} (Completeness)} \label{sec:complete}

    Our goal in this section is to find a short, sparse lattice vector $\mat x \mat B$, assuming the hypergraph $H$ satisfies case (1) of the QRDH problem.
    
    Notice that it will be sufficient to find a nonzero vector $\mat x \in \{-1,0,1\}^{M'}$ with at most 
    \[h^- = \lceil \alpha(1/r)^{d-1}M\rceil^q \leq h\]
    nonzero entries that satisfies $\mat x \mat C = \mat 0$. By construction, this implies that the nonzero entries in $\mat x \mat B$ come entirely from the single $M' \times M'$ identity matrix contained within $\mat B$, so $\mat x \mat B \in \{-1,0,1\}^{N'}$, and $\mat x \mat B$ has at most $h^- \leq h$ nonzero entries.
    
    The overall plan for finding $\mat x$ is to identify a small, compressing submatrix of $\mat C$, and then use the pigeonhole principle to argue that there exists a short, nonzero linear combination of the rows in this submatrix that sums to zero.

    \begin{lemma}
        Suppose that the hypergraph $H$ satisfies case (1) of the QRDH problem. Then assuming $n$ is sufficiently large, there exists a nonzero vector $\mat x \in \{-1,0,1\}^{M'}$ such that $\mat x\mat C = \mat 0$ and $\|\mat x\|_0 \leq h^-$.
    \end{lemma}

    \begin{proof}

    By assumption on the hypergraph $H$, there exists a set $V'$ of at most $N/r$ vertices that fully contains at least $\alpha(1/r)^{d-1}M$ hyperedges. Because the cardinality of a finite set is integral, the lower bound is in fact $\lceil\alpha(1/r)^{d-1}M\rceil$. Viewed in terms of the indicator matrix $\mat P$, we have a subset $E' \subseteq [M]$ of size at least $\lceil\alpha(1/r)^{d-1}M\rceil$, such that the nonzero entries of the rows indexed by $E'$ in $\mat P$ are supported on at most $N/r$ distinct columns. If the size of $E'$ strictly exceeds $\lceil\alpha(1/r)^{d-1}M\rceil$, then arbitrarily delete members of $E'$ until its size is exactly $\lceil\alpha(1/r)^{d-1}M\rceil$.

    Denote the $q$-fold Cartesian product of $E'$ with itself as $(E')^q$. Let the row-induced submatrix of $\mat C$ indexed by $(E')^q$ be $\mat C'$. By construction, $\mat C'$ has exactly $h^- = \lceil\alpha(1/r)^{d-1}M\rceil^q$ rows. Below we give an upper bound on the number of nonzero columns.

    \begin{claim}
        Assuming $n$ is sufficiently large, the row-induced submatrix $\mat C'$ has at most $h^-/\sqrt{n}$ nonzero columns.\footnote{This upper bound is quite loose, but it will be sufficient for our purposes.}
    \end{claim}

    \begin{proof}
        Recall that $\mat C = \big[\mat A \| \mat R \| \mat W\big]$, and write $\mat C' = \big[\mat A' \| \mat R' \| \mat W'\big]$, where $\mat A'$, $\mat R'$, and $\mat W'$ are (respectively) the row-induced submatrices of $\mat A$, $\mat R$, and $\mat W$ indexed by $(E')^q$. We show that each of $\mat A'$, $\mat R'$, and $\mat W'$ has $o(h^-/\sqrt{n})$ nonzero columns, which proves the claim for sufficiently large $n$.

        \paragraph{Matrix $\mat A'$.} Using the definition of Vandermonde fortified tensor products, we can write $\mat A \in \mathbb{Z}^{M' \times qM^q/t}$ in terms of a collection $\mathcal{S}$ of $qM^{q-1}$ subsets of $[M]^q$. Recall that $\mathcal{S}$ is defined as
        \[\mathcal{S} = \Big\{\{(e_1, \ldots, e_{\ell-1}, e_\ell', e_{\ell+1}, \ldots, e_q) : e_\ell' \in [M]\} \quad : \quad \ell \in [q] \text{ and } e_1, \ldots, e_{\ell-1}, e_{\ell+1}, \ldots, e_q \in [M]\Big\}.\]
        To construct $\mat A$, we take the horizontal concatenation of all matrices in $\{\mat A^{(S)} : S \in \mathcal{S}\}$, where each $\mat A^{(S)}$ is of width $M/t$ and has exactly $M$ nonzero rows, those being the rows indexed by $S$.

        So to bound the number of nonzero columns in $\mat A'$, it will be sufficient to bound the number of sets $S \in \mathcal{S}$ such that $(E')^q \cap S \neq \emptyset$, and then multiply the result by $M/t$. The number of such sets is exactly the number of choices for $\ell \in [q]$ and $e_1, \ldots, e_{\ell-1}, e_{\ell+1}, \ldots, e_q \in [M]$ such that there exists $e_\ell'$ with $(e_1, \ldots, e_{\ell-1}, e_\ell', e_{\ell+1}, \ldots, e_q) \in (E')^q$. There are $q$ choices for $\ell$, and there are $\lceil\alpha(1/r)^{d-1}M\rceil^{q-1}$ choices for $e_1, \ldots, e_{\ell-1}, e_{\ell+1}, \ldots, e_q$, due to the Cartesian product structure of $(E')^q$.

        Using $q \leq 2\log n/\log\log n$ and $t \geq n/(\alpha(1/r)^{d-1})$, the number of nonzero columns in $\mat A$ is at most
        \begin{align*}
            q\lceil\alpha(1/r)^{d-1}M\rceil^{q-1} \cdot M/t & \leq 2(\log n/\log\log n) \cdot \lceil\alpha(1/r)^{d-1}M\rceil^{q-1} \cdot M/(n/(\alpha(1/r)^{d-1})) \\
            & \leq \lceil\alpha(1/r)^{d-1}M\rceil^{q-1} \cdot M\alpha(1/r)^{d-1} \cdot (2\log n/(n\log\log n)) \\
            & \leq \lceil\alpha(1/r)^{d-1}M\rceil^q \cdot (2\log n/(n\log\log n)) \\
            & \leq o(h^-/\sqrt{n}). \\
        \end{align*}

        \paragraph{Matrix $\mat R'$.} To reason about $\mat R'$, consider the row-induced submatrix $\mat Q'$ of $\mat Q$ consisting of all rows indexed by $(E')^q$. Because $\mat Q = \mat P^{\otimes q}$, the number of nonzero columns in $\mat Q'$ is at most $(N/r)^q$. Now observe that every nonzero column of $\mat Q'$ maps to a set of $w$ nonzero columns in $\mat R'$. So using that
        \[w \leq 2\left(\frac{\alpha(1/r)^{d-1}M}{N/r}\right)^q/n,\]
        the total number of nonzero columns in $\mat R'$ is at most
        \begin{align*}
            (N/r)^q \cdot w & \leq 2(N/r)^q \cdot \left(\frac{\alpha(1/r)^{d-1}M}{N/r}\right)^q/n \\
            & \leq 2(\alpha(1/r)^{d-1}M)^q/n \\
            & \leq o(h^-/\sqrt{n}).
        \end{align*}
        
        \paragraph{Matrix $\mat W'$.} By construction, $\mat W$ has width $h/n$, which means that $\mat W'$ has at most $h/n$ nonzero columns. Using that
        \[h \leq (2\alpha(1/r)^{d-1}M)^q\]
        and
        \[h^- = \lceil \alpha(1/r)^{d-1}M\rceil^q,\]
        we know that $h \leq 2^qh^-$. Because $q \leq 2\log n/\log\log n$, the number of nonzero columns in $\mat W'$ is at most $2^{O(\log n/\log\log n)}h^-/n \leq o(h^-/\sqrt{n})$.
    \end{proof}

    So far, we have established that the number of rows in $\mat C'$ is $h^-$, and the number of nonzero columns in $\mat C'$ is at most $h^-/\sqrt{n}$. To finish the proof, we use a counting argument to show that there always exists a nonzero vector $\mat x' \in \{-1,0,1\}^{h^-}$ such that $\mat x' \mat C' = \mat 0$. After padding $\mat x'$ with zeros, we get a vector $\mat x \in \{-1,0,1\}^{M'}$ with at most $h^-$ nonzero entries such that $\mat x\mat C = \mat 0$.

    By construction, each entry of $\mat C'$ is a nonnegative integer of magnitude less than $3M^q$. So for all $\mat x'' \in \{0,1\}^{h^-}$, every coordinate of $\mat x'' \mat C'$ is nonnegative and less than $3M^qh^{-}$. Furthermore, all but a fixed set of at most $h^-/\sqrt{n}$ of the coordinates are zero. There are $2^{h^-}$ choices for such a vector $\mat x''$, but only the following number of possible values for the vector-matrix product $\mat x'' \mat C'$:
    \begin{align*}
        (3M^qh^{-})^{h^-/\sqrt{n}} & \leq (3M^q\lceil M\rceil^q)^{h^-/\sqrt{n}} \tag{$h^- = \lceil \alpha (1/r)^{d-1}M\rceil^q$, $\alpha < 1, r = \omega(1),$ and $d = \omega(1)$.} \\
        & \leq (3M^{2q})^{h^-/\sqrt{n}} \tag{$M$ is a positive integer.} \\
        & \leq (2^{O(2^{O(\sqrt{\log n}\log^4\log n)}q)})^{h^-/\sqrt{n}} \tag{$M = \exp{2^{O(\sqrt{\log n}\log^4\log n)}}$} \\
        & \leq (2^{2^{O(\sqrt{\log n}\log^4\log n)}})^{h^-/\sqrt{n}} \tag{$q \leq 2\log n/\log\log n$} \\
        & \leq 2^{h^-/n^{0.5 - o(1)}} \\
    \end{align*}
    So by the pigeonhole principle, assuming $n$ is sufficiently large, there must exist two distinct vectors $\mat x'', \mat x''' \in \{0,1\}^h$ such that $\mat x'' \mat C' = \mat x''' \mat C'$. Taking $\mat x' = \mat x'' - \mat x'''$ ensures that $\mat x'$ is nonzero, has all of its entries being in $\{-1,0,1\}$, and satisfies $\mat x' \mat C' = \mat 0$.
    \end{proof}

\subsection{Proof of Lemma \ref{lem:soundness} (Soundness)} \label{sec:sound}
In this subsection, we show that if the hypergraph $H$ satisfies case (2) of the QRDH problem, then every lattice vector $\mat x \mat B$ will have at least $2h$ nonzero entries. Our general proof strategy is to gradually enforce more and more structure on the possible vectors $\mat x \in \mathbb{Z}^{M'}$, until we rule out all possibilities with less than $2h$ nonzero entries.

Before giving the proof, we need to formalize the relationship between the nonzero entries of a vector and the columns of a matrix. To this end, let $\mat x \in \mathbb{Z}^{M'}$ be any vector, and let $\mat B \in \mathbb{Z}^{M' \times N'}$ be any matrix. We say that $\mat x$ \emph{implicates} a column $\mat c$ of $\mat B$ if the linear combination $\mat x \mat B$ assigns a nonzero coefficient to at least one row of $\mat B$ with a nonzero entry in column $\mat c$. Similarly, we say that $\mat x$ implicates a column $\mat c$ \emph{with multiplicity $x$} if the linear combination $\mat x\mat B$ assigns a nonzero coefficient to exactly $x$ rows of $\mat B$ with a nonzero entry in column $\mat c$.

We now massage the soundness condition from a statement about matrix $\mat B$ to a statement about matrix $\mat C$. Notice that any lattice vector $\mat x \mat B$ with $\|\mat x \mat B\|_0 < 2h$ must satisfy $\mat x \mat C = \mat 0$. Otherwise, because there are $2h$ horizontally concatenated copies of $\mat C$ in $\mat B$, the product $\mat x \mat B$ will have at least $2h$ nonzero entries. Also notice that, among all vectors $\mat x$ such that $\mat x \mat C = \mat 0$, the number of nonzero entries in $\mat x \mat B$ is exactly the number of nonzero entries in $\mat x$. This is because $\mat B$ contains a single $M' \times M'$ identity matrix. So to prove soundness, it will be sufficient to show that every nonzero vector $\mat x$ with $\mat x \mat C = \mat 0$ has at least $2h$ nonzero entries, or equivalently:

\begin{lemma}
    Suppose that the hypergraph $H$ satisfies case (2) of the QRDH problem. Then assuming $n$ is sufficiently large, for all nonzero $\mat x \in \mathbb{Z}^{M'}$ such that $\|\mat x\|_0 < 2h$, we have $\mat x \mat C \neq \mat 0$.
\end{lemma}

\begin{proof}

    As in the proof of completeness, recall that $\mat C = \big[\mat A \| \mat R \| \mat W\big]$. The bulk of the proof lies in a series of claims that make use of $\mat A$, $\mat R$, and $\mat W$ individually to prove that $\mat x \mat C \neq \mat 0$ under various assumptions about $\mat x$. At the end of the proof we consider all cases together, showing that they exhaust the set of all possible vectors $\mat x$ with $\|\mat x \|_0 < 2h$.
    
    First, we use $\mat W$ to argue that any nonzero vector $\mat x$ which is too sparse immediately has $\mat x \mat C \neq 0$. 

    \begin{claim} \label{claim:usingW}
        Every nonzero vector $\mat x \in \mathbb{Z}^{M'}$ such that $\|\mat x \|_0 < h/n$ satisfies $\mat x \mat W \neq \mat 0$.
    \end{claim}

    \begin{proof}
        Recall that $\mat W$ is simply a width $h/n$ reduced Vandermonde matrix, so by Lemma \ref{lem:vandindependent} every subset of less than $h/n$ rows is linearly independent. This means that any nontrivial linear combination $\mat x \mat W$ with less than $h/n$ nonzero coefficients is nonzero.
    \end{proof}

    Now we reason about the $(t, q)$-VF tensor product matrix $\big[\mat A \| \mat Q\big]$, from which $\big[\mat A \| \mat R\big]$ was constructed.

    \begin{claim} \label{claim:mincolQ}
        Let $\mat x \in \mathbb{Z}^{M'}$ be any vector such that $\|\mat x\|_0 \geq h/n$, and suppose that $\mat x \mat A = \mat 0$.
        Then assuming $n$ is sufficiently large, $\mat x$ implicates more than $N^q\delta$ distinct columns of $\mat Q$, where
        \[\delta = \left(\frac{(\alpha(1/r)^{d-1})^q}{n^2}\right)^{1/d}.\]
    \end{claim}

    \begin{proof}
        To prove the claim, we just use Theorem \ref{thm:vandtensor}, since the matrix $\big[\mat A\|\mat Q\big]$ is exactly the $(t,q)$-VF tensor product of $H$'s indicator matrix $\mat P$.
        
        First we verify that $H$, $\beta$, and $t$ satisfy the preconditions of the theorem.
        We know that $H$ is a hypergraph of arity $d$ with $N$ vertices and $M$ hyperedges. Because $H$ falls into case (2) of the QRDH problem, we also know that $H$ satisfies the expansion requirement. In particular, every vertex subset $V' \subseteq V$ fully contains at most $\left(\frac{\vert V'\vert}{N}\right)^dM + \beta M$ hyperedges, where $\beta = n^{-\Omega(\log\log n)}$. By our choice of parameters, we know that $t$ divides $M$. Below we show that $t \leq n^{O(1)}$; because $\beta = n^{-\Omega(\log\log n)}$, this shows that $1/t > \beta$ for sufficiently large $n$.
        \begin{align*}
            t \leq & \frac{2n}{\alpha(1/r)^{d-1}} \\
            \leq & \frac{4n}{\left(1/2^{\Theta(\sqrt{\log n})}\right)^{O(\sqrt{\log n}/\log^2\log n) - 1}} \tag{$\alpha = 0.5, r = 1/2^{\Theta(\sqrt{\log n})}, d = O(\sqrt{\log n}/\log^2\log n)$}\\
            \leq & \frac{4n}{2^{-O(\log n/\log^2\log n)}} \\
            \leq & n^{O(1)}
        \end{align*}

        Now to apply Theorem \ref{thm:vandtensor}, let $\mat x$ be any vector such that $\|\mat x\|_0 \geq h/n$, and assume that $\mat x \mat A = \mat 0$. The theorem tells us that, if $h/n > \frac{M^q\delta^d}{(1 - \beta t)^q}$ (and hence $\|\mat x\|_0 > \frac{M^q\delta^d}{(1 - \beta t)^q}$), then $\mat x$ implicates more than $N^q\delta$ distinct columns of $\mat Q$. All that remains is to verify the (strict) inequality:

        \begin{align*}
            \frac{M^q\delta^d}{(1 - \beta t)^q} & \leq \frac{M^q\cdot \frac{(\alpha(1/r)^{d-1})^q}{n^2}}{(1 - \beta t)^q} \tag{$\delta = \left(\frac{(\alpha(1/r)^{d-1})^q}{n^2}\right)^{1/d}$}\\
            & \leq M^q \cdot \frac{(\alpha(1/r)^{d-1})^q}{n^2(1 - \beta t)^q} \\
            & \leq M^q \cdot \frac{(\alpha(1/r)^{d-1})^q}{n^2(1 - n^{-\omega(1)})^{2\log n/\log\log n}} \tag{$\beta = n^{-\Omega(\log\log n)}$, $t \leq n^{O(1)}$, $q \leq 2\log n/\log\log n$}\\
            & < M^q \cdot \frac{(\alpha(1/r)^{d-1})^q}{0.5n^2} \tag{$n$ sufficiently large}\\
            & < (\alpha(1/r)^{d-1}M)^q/n \tag{$n$ sufficiently large} \\
            & \leq h/n \tag{$h \geq \lceil\alpha(1/r)^{d-1}M\rceil^q$}
        \end{align*}
    \end{proof}

    Below, we argue that if a vector $\mat x$ is sufficiently sparse and implicates too many distinct columns of $\mat Q$, then we have $\mat x \mat R \neq 0$.

    \begin{claim} \label{claim:implicateR}
        Let $\mat x \in \mathbb{Z}^{M'}$ be any vector such that $\|\mat x\|_0 < 2h$, and assume that $\mat x$ implicates more than
            \[N^q\left(\frac{(\alpha(1/r)^{d-1})^q}{n^2}\right)^{1/d}\]
            distinct columns of $\mat Q$.
        Then assuming $n$ is sufficiently large, we have $\mat x \mat R \neq \mat 0$.
    \end{claim}

    \begin{proof}
        We know that every row of $\mat Q$ has exactly $d^q$ nonzero entries, and that the linear combination $\mat x \mat Q$ assigns a nonzero coefficient to less than $2h$ rows of $\mat Q$. So by the averaging principle, there must exist a column $\mat Q_{\cdot, \mat v^*}$ of $\mat Q$ that is implicated by $\mat x$ with multiplicity at least $1$ and less than
        \[\frac{2h \cdot d^q}{N^q\left(\frac{(\alpha(1/r)^{d-1})^q}{n^2}\right)^{1/d}}.\]
        Fix this column $\mat Q_{\cdot, \mat v^*}$ for the rest of the proof. Manipulating the expression above, we argue that $\mat x$ implicates $\mat Q_{\cdot, \mat v^*}$ with multiplicity strictly less than $w$, assuming $n$ is sufficiently large.
        \begin{proposition} \label{claim:avgprinciple}
        Let $n$, $M$, and $N$ be positive integers, and let
        \begin{align*}
            \alpha & = 0.5 \\
            r & = 2^{\Theta(\sqrt{\log n})} \\
            \log n/\log\log n \leq q & \leq 2\log n/\log\log n \\
            1 \leq d & \leq O(\sqrt{\log n}/\log^2\log n) \\
            h & \leq (2\alpha(1/r)^{d-1}M)^q \\
            \left(\frac{\alpha(1/r)^{d-1}M}{N/r}\right)^q/n \leq w &
        \end{align*}
        Then assuming $n$ is sufficiently large,
        \[\frac{2h \cdot d^q}{N^q\left(\frac{(\alpha(1/r)^{d-1})^q}{n^2}\right)^{1/d}} < w.\]
    \end{proposition}
    \begin{proof}
        We defer the calculations to Appendix \ref{app:claim:implicateR}.
    \end{proof}

    By construction of $\mat R$, we know that the column $\mat Q_{\cdot, \mat v^*}$ maps to a column-induced submatrix $\mat R^{(\mat v^*)}$ of $\mat R$. Each row of $\mat R^{(\mat v^*)}$ is nonzero if and only if the corresponding entry of $\mat Q_{\cdot, \mat v^*}$ is nonzero. Additionally, each of the nonzero rows is a distinct row of width $w$ reduced Vandermonde matrix. So by Proposition \ref{claim:avgprinciple} (assuming $n$ is sufficiently large), $\mat x \mat R^{(\mat v^*)}$ must be a nontrivial linear combination of at least one, and strictly less than $w$, distinct rows of a width $w$ reduced Vandermonde matrix. By Lemma \ref{lem:vandindependent} we know that these rows are linearly independent, so $\mat x \mat R^{(\mat v^*)} \neq \mat 0$, and hence $\mat x \mat R \neq 0$.
    \end{proof}

    With Claims \ref{claim:usingW}, \ref{claim:mincolQ}, and \ref{claim:implicateR} at hand, the lemma follows essentially immediately. Let $\mat x \in \mathbb{Z}^{M'}$ be any nonzero vector such that $\|\mat x\|_0 < 2h$. By Claim \ref{claim:usingW}, we either have $\mat x \mat W \neq \mat 0$, in which case we are done, or we know that $\|\mat x \|_0 \geq h/n$. Assuming $\|\mat x\|_0 \geq h/n$, we know by Claim \ref{claim:mincolQ} that either $\mat x \mat A \neq \mat 0$, in which case we are again done, or $\mat x$ implicates at least
    \[N^q\left(\frac{(\alpha(1/r)^{d-1})^q}{n^2}\right)^{1/d}\]
    distinct columns of $\mat Q$. In this final case, by Claim \ref{claim:implicateR}, it must be that $\mat x \mat R \neq 0$. So in all cases, we have $\mat x \big[\mat A \| \mat R \| \mat W\big] \neq \mat 0$, which by definition of $\mat C$ implies that $\mat x \mat C \neq \mat 0$.
\end{proof}

\section{Hardness of Approximation for Quasi-Random Densest Sub-Hypergraph} \label{sec:qrdhproof}
In this section we prove Theorem \ref{thm:khotqrdh} by adapting the quasi-random PCPs of \cite{khot2006ruling, khot2016hardness}. We re-state the definition of QRDH, and Theorem \ref{thm:khotqrdh}, below.

\paragraph{Problem \ref{prob:qrdh}, Restated.}
    \emph{Given a hypergraph $H = (V, E)$ of arity $d$ with $\vert V\vert = N$ and $\vert E\vert = M$, along with a parameter $r \geq 1$ and parameters $\alpha, \beta \in [0, 1]$, distinguish between the following two cases:
    \begin{enumerate}
        \item There exists a vertex subset $V' \subseteq V$ of size at most $N/r$ that fully contains at least $\alpha(1/r)^{d-1} M$ hyperedges.
        \item Every vertex subset $V' \subseteq V$ fully contains at most $(\vert V'\vert/N)^d M + \beta M$ hyperedges.
    \end{enumerate}
    As shorthand, we refer to the problem as $(N, M, d, r, \alpha, \beta)$-QRDH. We refer to $r$ and $\alpha$ as the \emph{completeness parameters}, and $\beta$ as the \emph{soundness parameter}.} 

\paragraph{Theorem \ref{thm:khotqrdh}, Restated.}
    \emph{There is a deterministic $\exp{2^{O(\sqrt{\log n}\log^4\log n)}}$ time reduction from SAT instances of size $n$ to 
    \begin{align*}
        \big( & N = \exp{2^{O(\sqrt{\log n}\log^4\log n)}}, \\
        & M = \exp{2^{O(\sqrt{\log n}\log^4\log n)}}, \\
        & d = O(\sqrt{\log n}/\log^2\log n), \\
        & r = 2^{\Theta(\sqrt{\log n})}, \alpha = 0.5, \beta = n^{-\Omega(\log\log n)}\big)\textnormal{-QRDH}.
    \end{align*}}

\subsection{Overview}
Before presenting our adaptation in detail, we give some informal background on the existing quasi-random PCPs, and give a brief overview of how we modify and make use of them.

As originally stated, the result in \cite{khot2006ruling} is as follows:

\begin{theorem}[Theorem 1.9 from \cite{khot2006ruling}] \label{thm:khotpcp}
    For every $\varepsilon > 0$, there is an integer $d = O(1/\varepsilon \log(1/\varepsilon))$ such that the following holds for all sufficiently large $n$. There is a randomized $2^{O(n^\varepsilon)}$ time algorithm which takes as input a SAT instance of size $n$ and outputs the description of a PCP verifier, such that:
    \begin{enumerate}
        \item The proof for the verifier is of size $2^{O(n^\varepsilon)}$. We denote using $\Pi$ the set of all proof locations.
        \item The verifier uses $O(n^\varepsilon)$ random bits to choose a subset $Q \subseteq \Pi$ of proof locations, where $\vert Q\vert = d$.
        \item Suppose the starting SAT instance was satisfiable. Then there exists a subset $\Pi' \subseteq \Pi$ of at most half the locations in the proof such that
        \[\Pr_{Q}[Q \subseteq \Pi'] \geq (1 - O(1/d)) \cdot (1/2)^{d-1}.\]
        \item Suppose the starting SAT instance was unsatisfiable. Then for all subsets $\Pi' \subseteq \Pi$ of at most half the locations in the proof,
        \[\Pr_Q[Q \subseteq \Pi'] \leq (1/2)^d + 1/2^{20d}.\]
    \end{enumerate}
\end{theorem}

Notice that the \emph{query pattern} of the verifier is what changes depending on whether the starting SAT instance was satisfiable or not. This is why we do not specify the alphabet for the proof. Also notice that the algorithm which constructs the PCP is \emph{randomized}. As pointed out by Khot and Saket \cite{khot2016hardness}, the only reason this algorithm is randomized is because \cite{khot2006ruling} makes use of a randomized hardness of approximation result for the minimum distance of code (MDC) problem. Using any of the deterministic hardness results that appeared later \cite{cheng2009deterministic, austrin2014simple, bhattiprolu2025pcp}, the construction of the PCP verifier becomes fully deterministic. Khot and Saket \cite{khot2016hardness} also point out that item (4) from Theorem \ref{thm:khotpcp} can be replaced with a stronger condition:
\emph{\begin{enumerate}
    \item[4.] Suppose the starting SAT instance was unsatisfiable. Then for all $\zeta \in [0,1]$ and for all subsets $\Pi'$ containing a $\zeta$ fraction of the locations in the proof, we have
        \[\Pr_Q[Q \subseteq \Pi'] \leq \zeta^d + 1/2^{20d}.\]
\end{enumerate}}

To interpret Theorem \ref{thm:khotpcp} in terms of QRDH, construct a hypergraph $H$ as follows. Execute the $2^{O(n^\varepsilon)}$ time (deterministic) algorithm to construct the PCP verifier. For every proof location, add one vertex to the hypergraph. For every choice of randomness $\mat z \in \{0,1\}^{O(n^\varepsilon)}$, identify the corresponding size-$d$ subset $Q \subseteq \Pi$, and add a hyperedge to $H$ containing the $d$ vertices which correspond to $Q$.

There are $2^{O(n^\varepsilon)}$ proof locations, so $H$ has $N \leq 2^{O(n^\varepsilon)}$ vertices. There are $2^{O(n^\varepsilon)}$ total choices for the bit string $\mat z$, so $H$ has $M \leq 2^{O(n^\varepsilon)}$ hyperedges. By construction, each hyperedge is of arity $d = O(n^\varepsilon)$. 
Items (3) and (4) from Theorem \ref{thm:khotpcp} become:
\emph{\begin{enumerate}
    \item[3.] Suppose that the starting SAT instance was satisfiable. Then there exists a subset $\Pi'$ of exactly half the vertices in $H$ such that a $(1 - O(1/d)) \cdot (1/2)^{d-1}$ fraction of the hyperedges in $H$ have all of their endpoints in $\Pi'$.
    \item[4.] Suppose that the starting SAT instance was unsatisfiable. Then for all $\zeta \in [0,1]$ and for all subsets $\Pi'$ containing a $\zeta$ fraction of the vertices in $H$, the fraction of hyperedges in $H$ that have all of their endpoints in $\Pi'$ is at most $\zeta^d + 1/2^{20d}$.
\end{enumerate}}
We can thus re-interpret Theorem \ref{thm:khotpcp} as:

\begin{theorem}[\cite{khot2006ruling}] \label{thm:khotgraph}
    For every $\varepsilon > 0$, there is a deterministic $2^{O(n^\varepsilon)}$ time reduction from SAT instances of size $n$ to
\begin{align*}
    \big( & N = \exp{O(n^\varepsilon)}, \\
    & M = \exp{O(n^\varepsilon)}, \\
    & d = O(1/\varepsilon \log(1/\varepsilon)), \\
    & r = 2, \alpha = (1-O(1/d)), \beta = 1/2^{20d}\big)\textnormal{-QRDH},
\end{align*}
\end{theorem}

Khot and Saket \cite{khot2016hardness} adapt the quasi-random PCP in \cite{khot2006ruling} to show hardness of approximation for the bipartite expansion problem. Phrased in terms of QRDH, their quasi-random PCP allows us to (essentially) recover Theorem \ref{thm:khotgraph}, but with $r$ set to a larger power of two. Informally, if we use exactly the same reduction as \cite{khot2006ruling, khot2016hardness}, but adjust the parameters so that $d$ becomes superconstant, we immediately get a deterministic reduction from SAT instances of size $n$ to the following, where $\varepsilon > 0$ is any constant.
\begin{align*}
    \big( & N = \exp{2^{O(\sqrt{\log n}\log^4\log n)}}, \\
    & M = \exp{2^{O(\sqrt{\log n}\log^4\log n)}}, \\
    & d = O(\sqrt{\log n}/\log^2\log n), \\
    & r = 2^{\Theta(\sqrt{\log n})}, \alpha \geq 0.5, \beta = \varepsilon\big)\textnormal{-QRDH}.
\end{align*}

In particular, while we can take $\beta = 1/2^{20d}$ when $d$ is a constant, $\beta$ stops scaling inverse-exponentially with $d$ as soon as $d = \omega(1)$. To drive $\beta$ all the way down to $n^{-\Omega(\log\log n)}$, it turns out that we only need to modify the very first step in the quasi-random PCP reduction, which is to show sufficiently strong NP hardness of approximation for MDC. In Section \ref{sec:strongermdc}, we replace this NP hardness proof with an $n^{O(\log\log n)}$ time reduction that achieves an even stronger approximation gap. In Section \ref{sec:carrythrough}, we show that this improved approximation gap is inherited through the rest of the reduction, allowing us to take $\beta = n^{-\Omega(\log\log n)}$.

\subsection{Strong Hardness of Approximation for MDC} \label{sec:strongermdc}

Our starting point is the minimum distance of code problem. We define the \emph{relative weight} of a vector as the number of nonzero entries divided by the total number of entries. Similarly, we define the \emph{relative minimum distance} of a linear code as the minimum relative weight over all nonzero codewords. The \emph{relative number of zero entries} in a vector is the number of zero entries divided by the total number of entries.

\begin{problemnormal}[Gap Minimum Distance of Code (GapMDC)]
    Given a matrix $\mat G \in \mathbb{F}^{M \times N}$ over a finite field $\mathbb{F}$, along with two parameters $0 < \alpha < \beta < 1$, distinguish between the following two cases:
    \begin{enumerate}
        \item There exists a nonzero vector $\mat x \in \mathbb{F}^M$ such that $\mat x\mat G$ has relative weight at most $\alpha$.
        \item For all nonzero vectors $\mat x \in \mathbb{F}^M$, $\mat x \mat G$ has relative weight at least $\beta$.
    \end{enumerate}
    As shorthand, we refer to the problem as $(\mathbb{F}, N, M, \alpha, \beta)$-GapMDC. We assume that all rows of $\mat G$ are linearly independent; otherwise the problem is trivial.
\end{problemnormal}

NP hardness of the exact minimum distance of code problem was first shown by Vardy \cite{vardy2002intractability}. Constant factor hardness of approximation was then proved by Dumer, Micciancio, and Sudan \cite{dumer2003hardness}, but only under the assumption that NP $\not\subseteq$ RP. Cheng and Wan \cite{cheng2009deterministic} were the first to de-randomize the hardness result, with later simplifications and refinements by Austrin and Khot \cite{austrin2014simple}, Micciancio \cite{micciancio2014locally}, and Bhattiprolu, Guraswami, and Ren \cite{bhattiprolu2025pcp}. We use the following deterministic hardness result:

\begin{theorem}[\cite{bhattiprolu2025pcp}] \label{thm:baseMDC}
    There exist constants $0 < \alpha < \beta < 1$ such that there is a deterministic $n^{O(1)}$ time reduction from SAT instances of size $n$ to $(\mathbb{F}_2, N, M, \alpha, \beta)$-GapMDC, where $N = n^{O(1)}$ and $M = n^{O(1)}$.
\end{theorem}

Using techniques similar to \cite{khot2006ruling}, we amplify the gap between the ``yes'' and ``no'' cases significantly.

\begin{theorem} \label{thm:bigMDC}
    There is a deterministic $n^{O(\log\log n)}$ time reduction from SAT instances of size $n$ to
    \[(\mathbb{F}_{2^\lambda}, N, M, 0.5, (1-n^{-\Omega(\log\log n)}))\text{-GapMDC},\]
    where $N = n^{\Theta(\log\log n)}$, $M = n^{O(\log\log n)}$, $N \leq 2^\lambda \leq N^2$, and $\lambda$ is a power of $2$.
\end{theorem}

The rest of this sub-section is dedicated to proving Theorem \ref{thm:bigMDC}.

\paragraph{Vandermonde Matrices.} A key subroutine for all steps of our MDC reduction will be to construct linear combinations of vectors using coefficients that come from a Vandermonde matrix.

\begin{definition}[Vandermonde Matrix]
    Given a finite field $\mathbb{F}$ and positive integers $a < \vert\mathbb{F}\vert, b < \vert\mathbb{F}\vert$, an $a \times b$ Vandermonde matrix over $\mathbb{F}$ is a matrix $\mat V \in \mathbb{F}^{a \times b}$ defined as follows. Associate each column index $j \in [b]$ with a distinct nonzero member $c(j) \in \mathbb{F}$. Now set $\mat V_{i, j} = c(j)^i$ for all $i, j$.
\end{definition}

Note that these matrices can be constructed efficiently and deterministically. We use the following lemma to reason about linear combinations that use Vandermonde matrices.

\begin{lemma} \label{lem:vandcombo}
    Let $\mathbb{F}$ be a finite field, and let $a, b, c$ be positive integers such that $a < \vert \mathbb{F}\vert$ and $b < \vert\mathbb{F}\vert$. Let $\mat V$ be an $a \times b$ Vandermonde matrix over $\mathbb{F}$, and let $\mat C \in \mathbb{F}^{c \times a}$ be any matrix. Then for all vectors $\mat x \in \mathbb{F}^c$ such that $\mat x \mat C \neq \mat 0$, it holds that $\mat x \mat C \mat V$ has at most $a$ zero entries.
\end{lemma}

\begin{proof}
    Let $\mat x \in \mathbb{F}^c$ be any vector such that $\mat x \mat C \neq \mat 0$, and let $\mat w = \mat x \mat C \in \mathbb{F}^a$. Then $\mat x \mat C \mat V = \mat w \mat V$ is simply the evaluation of nonzero polynomial of degree at most $a$ on $b$ distinct points. Any such polynomial can have at most $a$ roots, so $\mat x \mat C \mat V$ has at most $a$ zero entries.
\end{proof}

\paragraph{Expander Graphs.} As in \cite{khot2006ruling}, we make critical use of expander graphs. We use the following explicit construction:

\begin{lemma}[See e.g. \cite{moshkovitz2008two} Lemma 5.2] \label{lem:basicexpander}
    There is a constant $\kappa < 1$ and a function $T(\Delta) = \Theta(\Delta)$ such that the following holds for all positive integers $n$ and $\Delta$. There is a deterministic $O((n\Delta)^{O(1)})$ time algorithm which outputs an undirected (multi)graph $G = (V, E)$ where $\vert V\vert = n$, every vertex is of the same degree $T(\Delta)$, and the adjacency matrix for $G$ has its second largest eigenvalue being of magnitude at most $(T(\Delta))^\kappa$.
\end{lemma}

A well-known result is that random walks along expander graphs have good mixing properties with respect to vertex subsets. We always identify a walk with the (ordered) set of vertices that it visits.

\begin{lemma}[See e.g. \cite{vadhan2012pseudorandomness} Theorem 4.17] \label{lem:basicwalk}
    Let $G = (V, E)$ be an undirected (multi)graph with $n$ vertices, each of degree $d$. Suppose that the magnitude of the second largest eigenvalue of $G$'s adjacency matrix is at most $\rho$. Then for all positive integers $r$ and all subsets $S \subseteq V$, the probability that a random walk $(v_1, \ldots, v_r)$ along $G$ has its vertices contained within $S$ is at most $(\vert S\vert/n + \rho/d)^r$.
\end{lemma}

Combining Lemmas \ref{lem:basicexpander} and \ref{lem:basicwalk}, we have the following corollary.

\begin{corollary} \label{cor:getexpander}
    For all $\delta \in (0,1)$, there exists a positive integer $d = O((1/\delta)^{O(1)})$ such that the following holds for all positive integers $n$. There is an $O((n/\delta)^{O(1)})$ time deterministic algorithm which outputs a $d$-regular graph $G = (V, E)$ on $n$ vertices such that, for all positive integers $r$ and all subsets $S \subseteq V$, the probability that a random walk $(v_1, \ldots, v_r)$ along $G$ has its vertices contained within $S$ is at most $(\vert S\vert/n + \delta)^{r}$.
\end{corollary}

\paragraph{Combining Expanders with Vandermonde Matrices.} We combine Lemma \ref{lem:vandcombo} with Corollary \ref{cor:getexpander} to get a useful subroutine for our MDC reduction.

\begin{lemma} \label{lem:expandervand}
    Let $w \geq r$ be positive integers, let $\delta \in (0,1)$, let $\mathbb{F}$ be a finite field of size more than $w$, and let $N$ and $M$ be positive integers. There is an $O((NMw\log\vert\mathbb{F}\vert)^{O(1)}(1/\delta)^{O(r)})$ time deterministic algorithm which takes as input a matrix $\mat G \in \mathbb{F}^{M \times N}$ and outputs a matrix $\mat G' \in \mathbb{F}^{M \times N'}$ such that $w \leq N' \leq O(Nw(1/\delta)^{O(r)})$, and such that the following holds. Let $D$ be the relative minimum distance of the code $\{\mat x \mat G : \mat x \in \mathbb{F}^M\}$, and let $D'$ be the relative minimum distance of the code $\{\mat x \mat G' : \mat x \in \mathbb{F}^M\}$. Then
    \begin{enumerate}
        \item $D' \leq rD$, and \label{item:expanderupper}
        \item $D' \geq 1 - \left((1 - D + \delta)^r + r/w\right)$. \label{item:expanderlower}
    \end{enumerate}
\end{lemma}

\begin{proof}
    Invoke Corollary \ref{cor:getexpander} with parameter $\delta$ to get a graph $G = (V, E)$ on $N$ vertices, where each vertex is of degree $d = O((1/\delta)^{O(1)})$. Assign each vertex $v \in V$ to a distinct column $\mat c^{(v)}$ of $\mat G$. Now do the following for every walk $(v_1, \ldots, v_r)$ on $G$:
    \begin{enumerate}
        \item Horizontally concatenate the columns $\mat c^{(v_1)}, \ldots, \mat c^{(v_r)}$ to get a matrix $\mat C \in \mathbb{F}^{M \times r}$.
        \item Let $\mat V \in \mathbb{F}^{r \times w}$ be a Vandermonde matrix; such a matrix exists because $\vert \mathbb{F}\vert > w \geq r$ by assumption.
        \item Let $\mat G^{(v_1, \ldots, v_r)} = \mat C \mat V$.
    \end{enumerate}
    The final matrix $\mat G'$ is the horizontal concatenation of all $\mat G^{(v_1, \ldots, v_r)}$ formed as above.
    
    The number of rows in $\mat G'$ is equal to the number of rows in $\mat G$, so $\mat G'$ is of height $M$ as required. The number of columns $N'$ in $\mat G'$ is $w$ times the number of length-$r$ walks in the graph $G$; because the graph is $O((1/\delta)^{O(1)})$-regular, we have $w \leq N' \leq O(Nw(1/\delta)^{O(r)})$. By Corollary \ref{cor:getexpander}, the graph $G$ can be constructed deterministically in time $O((N/\delta)^{O(1)})$. Combining all of this with the fact that arithmetic operations over $\mathbb{F}$ can be performed in time $(\log\vert\mathbb{F}\vert)^{O(1)}$, the time to construct $\mat G'$ is $O((NMw\log\vert\mathbb{F}\vert)^{O(1)} \cdot (1/\delta)^{O(r)})$. All that remains is to argue the relative minimum distance properties.

    \begin{claim} \label{claim:upper}
        $D' \leq rD$.
    \end{claim}
    \begin{proof}
        Let $\mat x \in \mathbb{F}^M$ be any nonzero vector such that $\mat x \mat G$ has relative weight at most $D$; such a vector exists by definition of $D$. We argue that $\mat x \mat G'$ has relative weight at most $rD$, which implies that $D' \leq rD$.

        Notice that every length-$r$ walk $(v_1, \ldots, v_r)$ in the graph $G$ such that $\mat x \mat c^{(v_i)} = 0$ for all $i \in [r]$ maps to a column-induced submatrix $\mat G^{(v_1, \ldots, v_r)}$ such that $\mat x \mat G^{(v_1, \ldots, v_r)} = \mat 0$. Therefore, all of the nonzero entries in $\mat x \mat G'$ must correspond to length-$r$ walks $(v_1, \ldots, v_q)$ on the graph $G$ such that there exists $i \in [q]$ with $\mat x \mat c^{(v_i)} \neq 0$. Because every submatrix $\mat G^{(v_1, \ldots, v_q)}$ has the same width, to upper bound the relative weight of $\mat x \mat G'$ it will be sufficient to upper bound the fraction of length-$r$ walks on $G$ satisfying this not-all-zeros property.

        We know that at most $DN$ of the columns in $\mat G$ have a nonzero dot product with $\mat x$. By the regularity of graph $G$, each of these columns corresponds to a vertex that appears in an $r/N$ fraction of the walks. So the total number of walks with at least one vertex corresponding to this set of columns is at most $rD$.
    \end{proof}

    \begin{claim} \label{claim:lower}
        $D' \geq 1 - \left((1 - D + \delta)^r + r/w\right)$.
    \end{claim}
    \begin{proof}
        Let $\mat x \in \mathbb{F}^M$ be such that $\mat x \mat G \neq \mat 0$. By definition of $D$, we know that $\mat x \mat G$ has relative weight at least $D$. We now argue that the relative number of zero entries in $\mat x \mat G'$ is at most
        \[(1 - D + \delta)^r + r/w,\]
        which is sufficient to prove the claim.

        As described previously, every length-$r$ walk $(v_1, \ldots, v_r)$ in the graph $G$ maps to a column-induced submatrix $\mat G^{(v_1, \ldots, v_r)}$ of $\mat G'$. For each walk $(v_1, \ldots, v_r)$, if there exists $i \in [r]$ such that $\mat x \mat c^{(v_i)} \neq 0$, then by Lemma \ref{lem:vandcombo} the relative number of zero entries in $\mat x \mat G^{(v_1, \ldots, v_r)}$ is at most $r/w$. So the relative number of zero entries in $\mat x \mat G'$ contributed by walks of this type is at most $r/w$.

        All the remaining zero entries come from walks $(v_1, \ldots, v_r)$ such that $\mat x \mat c^{(v_i)} = 0$ for all $i \in [r]$. 
        We know that at most a $1 - D$ fraction of the vertices in graph $G$ correspond to columns of $\mat G$ that have a dot product of zero with $\mat x$. By Corollary \ref{cor:getexpander}, the probability that a random length-$r$ walk on $G$ has all of its vertices contained within this subset is at most $(1 - D + \delta)^r.$ Thus, the fraction of submatrices $\mat G^{(v_1, \ldots, v_r)}$ such that $\mat x \mat G^{(v_1, \ldots, v_r)} = \mat 0$ is at most $(1 - D + \delta)^r.$ Because every submatrix is of the same width $w$, this quantity also serves as an upper bound on the relative number of additional zero entries, giving a total upper bound of
        \[(1 - D + \delta)^r + r/w\]
        on the relative number of zero entries.
    \end{proof}
    Combining Claims \ref{claim:upper} and \ref{claim:lower}, we have proved the lemma.
\end{proof}

\paragraph{An Initial Strengthening.} Over a sufficiently large extension field of $\mathbb{F}_2$, we can amplify the completeness and soundness parameters to any constants in $(0, 1)$, while only incurring a polynomial blow-up in the instance size and running time.

\begin{lemma} \label{lem:anyconstants}
    There is a computable function $f(\alpha, \beta)$ such that the following holds. Let $0 < \alpha < \beta < 1$ be any constants, let $\mathbb{F}$ be any extension field of $\mathbb{F}_2$ of size more than $f(\alpha, \beta)$, and let $n$ be any positive integer. There is a deterministic $O((n\log\vert\mathbb{F}\vert)^{O(1)})$ time reduction from SAT instances of size $n$ to $(\mathbb{F}, N, M, \alpha, \beta)$-GapMDC, where $N = n^{O(1)}$ and $M = n^{O(1)}$.
\end{lemma}

\begin{proof}
    Let $\alpha'$ and $\beta'$ be, respectively, the completeness and soundness parameters guaranteed by Theorem \ref{thm:baseMDC}. We know that $0 < \alpha' < \beta' < 1$. Choose any $0 < \alpha < \beta < 1$ to be the target completeness and soundness parameters. Now let $q$ be the smallest positive integer such that
    \[\left\lceil\frac{2\log_2(2/(1-\beta))}{(\beta')^q}\right\rceil \leq \left\lfloor \frac{\alpha}{(\alpha')^q}\right\rfloor,\]
    and then let $q'$ be the smallest positive integer satisfying
    \[\frac{2\log_2(2/(1-\beta))}{(\beta')^q}\leq q' \leq \frac{\alpha}{(\alpha')^q}.\]
    A solution always exists by our assumptions on $\alpha, \beta, \alpha',$ and $\beta'$. We define the function $f(\alpha, \beta)$ as $f(\alpha, \beta) = q'\lceil 2/(1-\beta)\rceil$.

    To start the reduction, let $\mathbb{F}$ be any extension field of $\mathbb{F}_2$ of size more than $f(\alpha, \beta)$. Apply Theorem \ref{thm:baseMDC} to the starting SAT instance to get a matrix $\mat G \in \mathbb{F}_2^{M \times N}$, where $N = n^{O(1)}$ and $M = n^{O(1)}$. If the SAT instance was satisfiable, then the relative minimum distance of the code $\{\mat x \mat G : \mat x \in \mathbb{F}_2^M\}$ is at most $\alpha'$, and otherwise it is at least $\beta'$. Notice that we can directly embed $\mat G$ into any extension field of $\mathbb{F}_2$ while maintaining the minimum distance properties, so assume that we embed $\mat G$ into $\mathbb{F}$.

    We modify the matrix $\mat G$ via two operations to get matrices $\mat G'$ and $\mat G^*$. Matrix $\mat G^*$ is the generator matrix for our final code. First, we compute $\mat G' = \mat G^{\otimes q}$, i.e. the $q$-fold tensor product of $\mat G$ with itself. Because the relative minimum distance of a code is multiplicative under tensoring, the relative minimum distance of the code $\mathcal{C}' = \{\mat x\mat G' : \mat x \in \mathbb{F}^{M^q}\}$ is at most $(\alpha')^q$ if the starting SAT instance was satisfiable, and otherwise it is at least $(\beta')^q$. Because $q$ is independent of $n$, the size of $\mat G'$ is polynomial in $n$, and the time to construct $\mat G'$ is\footnote{The $(\log \vert\mathbb{F}\vert)^{O(1)}$ factor comes from performing arithmetic operations over $\mathbb{F}$.} $O((n\log\vert\mathbb{F}\vert)^{O(1)})$.

    Now apply the reduction from Lemma \ref{lem:expandervand} to $\mat G'$, where we set $w = q'\lceil 2/(1-\beta)\rceil$, $r = q'$, and $\delta = (\beta')^q/2$; the lemma applies because $\vert \mathbb{F}\vert > f(\alpha, \beta) = w$. Let the resulting matrix to be $\mat G^*$. Because $w, r,$ and $\delta$ are independent of $n$, the time to construct $\mat G^*$ is $O((n\log\vert\mathbb{F}\vert)^{O(1)})$, and $\mat G^*$ is of size $n^{O(1)}$. All that remains is to bound the relative minimum distance of the code $\mathcal{C}^* = \{\mat x \mat G^* : \mat x \in \mathbb{F}^{M^q}\}$.

    \begin{claim} \label{claim:anyconstyes}
        Suppose that the starting SAT instance was satisfiable, meaning the relative minimum distance of $\mathcal{C}'$ is at most $(\alpha')^q$. Then the relative minimum distance of $\mathcal{C}^*$ is at most $\alpha$.
    \end{claim}

    \begin{proof}
        By Lemma \ref{lem:expandervand}, we know that the relative minimum distance of $\mathcal{C}^*$ is at most $q'$ times the relative minimum distance of $\mathcal{C}'$. Using that $q' \leq \frac{\alpha}{(\alpha')^q}$, the relative minimum distance of $\mathcal{C}^*$ is bounded as
        \[q'(\alpha')^q \leq \frac{\alpha}{(\alpha')^q}(\alpha')^q = \alpha.\]
    \end{proof}

    \begin{claim}
        Suppose that the starting SAT instance was not satisfiable, meaning the relative minimum distance of $\mathcal{C}'$ is at least $(\beta')^q$. Then the relative minimum distance of $\mathcal{C}^*$ is at least $\beta$.
    \end{claim}

    \begin{proof}
        Let the relative minimum distance of $\mathcal{C}'$ be $D'$, and the relative minimum distance of $\mathcal{C}^*$ be $D^*$. Again using Lemma \ref{lem:expandervand}, we know that
        \[D^* \geq 1 - \left((1 - D' + \delta)^r + r/w\right).\]
        Substituting our choice of parameters, this becomes
        \begin{align*}
            D^* & \geq 1 - \left((1 - D' + \delta)^{q'} + q'/w\right) \tag{$r = q'$} \\
            & \geq 1 - \left((1 - (\beta')^q + (\beta')^q/2)^{q'} + q'/w\right) \tag{$D' \geq (\beta')^q$, $\delta = (\beta')^q/2$} \\
            & \geq 1 - \left((1 - (\beta')^q/2)^{(2/(\beta')^q)\log_2(2/(1-\beta))} + q'/w\right) \tag{$q' \geq \frac{2\log_2(2/(1-\beta))}{(\beta')^q}$} \\
            & \geq 1 - \left((1/2)^{\log_2(2/(1-\beta))} + q'/w\right) \tag{$(\beta')^q/2 \in (0,1)$} \\
            & \geq 1 - \left(\frac{1 - \beta}{2} + q'/w\right) \\
            & \geq 1 - \left(\frac{1 - \beta}{2} + \frac{1 - \beta}{2}\right) \tag{$w = q'\lceil 2/(1-\beta)\rceil$} \\
            & \geq \beta
        \end{align*}
    \end{proof}
    Combining the above claims, we achieve the desired minimum distance properties for $\mat G^*$.
\end{proof}

\paragraph{Making $\alpha$ Sub-Constant.} We now describe a self-reduction for MDC that allows us to improve the parameter $\alpha$ significantly, while keeping $\beta$ at a fixed constant.

\begin{lemma} \label{lem:inversepolylog}
    There is a computable function $f(\beta)$ such that the following holds. Let $c > 0$ and $\beta \in (0,1)$ be any constants, let $\mathbb{F}$ be any extension field of $\mathbb{F}_2$ of size more than $f(\beta)$, and let $n$ be any positive integer. There is a deterministic $O((n\log\vert\mathbb{F}\vert)^{O(\log\log n)})$ time reduction from SAT instances of size $n$ to $(\mathbb{F}, N, M, 1/(\log n)^c, \beta)$-GapMDC, where $N = n^{O(\log\log n)}$ and $M = n^{O(\log\log n)}$.
\end{lemma}

\begin{proof}
    Fix any constants $c > 0$ and $\beta \in (0,1)$, and solve for the smallest positive integer $q$ such that
    \[(1 - \beta^2/2)^q \leq \frac{1-\beta}{2};\]
    a solution always exists because $\beta \in (0,1)$. Let $\alpha = 1/(2q)$, and set $f(\beta) = \max(f'(\alpha, \beta), q\lceil 2/(1-\beta)\rceil)$, where $f'(\alpha, \beta)$ is the function from Lemma \ref{lem:anyconstants}.

    To start the reduction, let $\mathbb{F}$ be any extension field of $\mathbb{F}_2$ of size more than $f(\beta)$, and apply Lemma \ref{lem:anyconstants} to the starting SAT instance with this field $\mathbb{F}$ and parameters $\alpha$ and $\beta$. This gives, deterministically in $(n\log\vert \mathbb{F}\vert)^{O(1)}$ time, a matrix $\mat G \in \mathbb{F}^{M \times N}$ where $N = n^{O(1)}$ and $M = n^{O(1)}$.

    From here, we use a recursive procedure to construct the generator matrix for our final GapMDC instance. To this end, set $L = c\log\log n$, and define $\mat G^{(1)} = \mat G$. Set parameters $\delta = \beta^2/2$, $r = q$, and $w = q\lceil 2/(1-\beta)\rceil$. Do the following for $\ell = 2$ up to $\ell = L$:
    \begin{enumerate}
        \item Let $\mat A^{(\ell)} = \mat G^{(\ell-1)} \otimes \mat G$. In other words, $\mat A^{(\ell)}$ is the tensor product of the previous matrix $\mat G^{(\ell-1)}$ with the starting matrix $\mat G$.
        \item Invoke the algorithm from Lemma \ref{lem:expandervand} on matrix $\mat A^{(\ell)}$ with parameters $\delta, r,$ and $w$; by our choice of $f(\beta)$, $\mathbb{F}$ is large enough to apply the lemma. The resulting matrix is $\mat G^{(\ell)}$.
    \end{enumerate}
    Our final GapMDC instance will be with respect to the matrix $\mat G^{(L)}$.

    We know that the starting matrix $\mat G$ is of size $n^{O(1)}$, and that the width of $\mat G^{(\ell)}$ is only a constant factor larger than the width of $\mat A^{(\ell)}$. So each matrix $\mat G^{(\ell)}$ is larger than the previous matrix $\mat G^{(\ell-1)}$ by the same $n^{O(1)}$ factor. Because there are $O(\log\log n)$ iterations, the size of the final matrix $\mat G^{(L)}$ is thus $n^{O(\log\log n)}$. The construction is deterministic and proceeds in $O((n\log\vert\mathbb{F}\vert)^{O(\log\log n)})$ time.\footnote{As before, the dependence on $\log \vert \mathbb{F}\vert$ comes from performing arithmetic operations over $\mathbb{F}$.} All that remains is to bound the relative minimum distance of the code $\mathcal{C}^{(L)} = \{\mat x \mat G^{(L)} : \mat x \in \mathbb{F}^{M^L}\}$ in terms of the original code $\mathcal{C} = \{\mat x \mat G : \mat x \in \mathbb{F}^M\}$.

    \begin{claim}
        Suppose that the starting SAT instance was satisfiable, meaning that the relative minimum distance of $\mathcal{C}$ is at most $\alpha$. Then the relative minimum distance of $\mathcal{C}^{(L)}$ is at most $1/(\log n)^c$.
    \end{claim}
    \begin{proof}
        We give a proof by induction, showing that for all $\ell \in [L]$, the relative minimum distance of the code $\mathcal{C}^{(\ell)} = \{\mat x \mat G^{(\ell)} : \mat x \in \mathbb{F}^{M^\ell}\}$ is at most $1/2^\ell$. The base case $\ell = 1$ follows immediately because $\alpha$ is necessarily at most $1/2$, and the final case (when $L = c\log\log n$) is sufficient to prove the claim because $1/2^L = 1/(\log n)^c$.

        All that remains is the inductive case. To this end, fix $\ell \in [2, L]$, and suppose that $\mathcal{C}^{(\ell-1)}$ has relative minimum distance at most $1/2^{\ell-1}$. Because the relative minimum distance is multiplicative under tensoring, we know that the code $\{\mat x \mat A^{(\ell)} : \mat x \in \mathbb{F}^{(n')^\ell}\}$ has relative minimum distance at most $\alpha/2^{\ell-1}$. By Lemma \ref{lem:expandervand}, the relative minimum distance increases by at most a factor of $q$ when converting from this code to $\mathcal{C}^{(\ell)}$, giving an upper bound of
        \[\alpha q/2^{\ell-1} = (1/2)/2^{\ell-1} = 1/2^\ell.\]
    \end{proof}

    \begin{claim}
        Suppose that the starting SAT instance was not satisfiable, meaning that the relative minimum distance of $\mathcal{C}$ is at least $\beta$. Then the relative minimum distance of $\mathcal{C}^{(L)}$ is at least $\beta$.
    \end{claim}
    \begin{proof}
        As in the previous claim, we give a proof by induction, showing that relative minimum distance of $\mathcal{C}^{(\ell)} = \{\mat x \mat G^{(\ell)} : \mat x \in \mathbb{F}^{M^\ell}\}$ is at most $\beta$. The base case $\ell = 1$ is immediate, so let $\ell \in [2, L]$. Assuming that the relative minimum distance of $\mathcal{C}^{(\ell-1)}$ is at least $\beta$, we know that the code $\{\mat x \mat A^{(\ell)} : \mat x \in \mathbb{F}^{M^\ell}\}$ has relative minimum distance at least $\beta^2$. Applying Lemma \ref{lem:expandervand}, the relative minimum distance of $\mathcal{C}^{(\ell)}$ is thus at least
        \begin{align*}
            1 - \left((1 - \beta^2 + \delta)^r + r/w\right) & \geq 1 - \left((1 - \beta^2/2)^q + q/w\right) \tag{$\delta = \beta^2/2$, $r = q$} \\
            & \geq 1 - \left(\frac{1-\beta}{2} + q/w\right) \tag{$(1 - \beta^2/2)^{q} \leq \frac{1-\beta}{2}$} \\
            & \geq 1 - \left(\frac{1-\beta}{2} + \frac{1-\beta}{2}\right) \tag{$w = q\lceil 2/(1-\beta)\rceil$} \\
            & \geq \beta
        \end{align*}
    \end{proof}
    Combining the two claims, we have that $\mathcal{C}^{(L)}$ satisfies the required distance properties.
\end{proof}

\paragraph{Superconstant Length Expander Walks.} We finish by using a single iteration of Lemma \ref{lem:expandervand}, where the walk length is $r = \log n\log\log n$. Recall Theorem \ref{thm:bigMDC}:

\paragraph{Theorem \ref{thm:bigMDC}, Restated.}
    \emph{There is a deterministic $n^{O(\log\log n)}$ time reduction from SAT instances of size $n$ to
    \[(\mathbb{F}_{2^\lambda}, N, M, 0.5, (1-n^{-\Omega(\log\log n)}))\text{-GapMDC},\]
    where $N = n^{\Theta(\log\log n)}$, $M = n^{O(\log\log n)}$, $N \leq 2^\lambda \leq N^2$, and $\lambda$ is a power of $2$.}

\begin{proof}
    Let $\mathbb{F} = \mathbb{F}_{2^\lambda}$ be an extension field of $\mathbb{F}_2$ such that $\lambda \in [2\log n\log\log n, 4\log n\log\log n)$ is a power of two. (At the very end of the proof, we will cast from this field to an even larger field to get our final generator matrix.)
    
    Apply Lemma \ref{lem:inversepolylog} to the starting SAT instance with this field $\mathbb{F}$ and parameters $c = 2$ and $\beta = 0.75$. Assuming $n$ is sufficiently large, we have that $\vert \mathbb{F}\vert$ is greater than the function $f(\beta)$ from the lemma statement. This gives, deterministically in time $n^{O(\log\log n)}$, a matrix $\mat G \in \mathbb{F}^{M \times N}$, where $N = n^{O(\log\log n)}$ and $M = n^{O(\log\log n)}$. If the starting SAT instance was satisfiable, then the relative minimum distance $D$ of the code $\mathcal{C} = \{\mat x \mat G : \mat x \in \mathbb{F}^M\}$ is at most $1/\log^2 n$, and otherwise it is at least $0.75$.

    Now invoke the algorithm from Lemma \ref{lem:expandervand} on $\mat G$ with parameters $r = \log n \log\log n$, $w = n^{\log\log n}$, and $\delta = 0.25$; the lemma applies because $\vert \mathbb{F}\vert \geq n^{2\log\log n} > w$ whenever $n$ is sufficiently large. The algorithm is deterministic and runs in time $n^{O(\log\log n)}$. Denote the resulting matrix as $\mat G' \in \mathbb{F}^{M' \times N'}$. We know that $n^{\Omega(\log\log n)} \leq w \leq N' \leq n^{O(\log\log n)}$ and $M' = n^{O(\log\log n)}$. Let the code generated by $\mat G'$ be $\mathcal{C}' = \{\mat x \mat G' : \mat x \in \mathbb{F}^{M'}\}$, and denote by $D'$ the relative minimum distance of this code.
    
    If $\mathcal{C}$ has relative minimum distance $D \leq 1/\log^2 n$, then we know by Lemma \ref{lem:expandervand} that the relative minimum distance of $\mathcal{C}'$ is at most $rD \leq (\log n\log\log n)/\log^2 n < 0.5$, assuming $n$ is sufficiently large. On the other hand, if the relative minimum distance of $\mathcal{C}$ is at least $0.75$, then we know by Lemma \ref{lem:expandervand} that the code $\mathcal{C}'$ has relative minimum distance at least
    \begin{align*}
        1 - ((1 - D + \delta)^r + r/w) & = 1 - ((0.5)^{\log n\log\log n} + \log n\log\log n/(n^{\log\log n})) \\
        & = 1 - n^{-\Omega(\log\log n)}
    \end{align*}

    Now we adjust the size of $\mathbb{F}$ and the width of $\mat G'$ to fit the theorem's requirements. If the width of $\mat G'$ is less than $2^\lambda$, then horizontally concatenate copies of $\mat G'$ until the width of the resulting matrix $\mat G^*$ is at least $2^\lambda$. Otherwise, take $\mat G^* = \mat G'$. Denote by $N^*$ the width of $\mat G^*$, and notice that $N^* = n^{\Theta(\log\log n)}$. Then, solve for the unique positive integer $x$ such that $N^* \leq 2^{2^x\lambda} < (N^*)^2$, and cast $\mat G^*$ from $\mathbb{F}$ to the extension field $\mathbb{F}_{2^{2^x\lambda}}$. This is the generator matrix for our final GapMDC instance. By construction, $\mat G^*$ inherits the relative minimum distance properties of $\mat G'$, and the final $\lambda^* = 2^x\lambda$ is still a power of two.
\end{proof}

\subsection{From MDC to Quasi-Random Densest Sub-Hypergraph.} \label{sec:carrythrough}

With our hardness result for MDC at hand, the rest of the reduction to QRDH proceeds identically to the reductions described in \cite{khot2006ruling} and \cite{khot2016hardness}, just with a different choice of parameters. The next step is to reduce to a special type of CSP.

\begin{problemnormal}[HomAlgCSP]
    A \emph{Homogeneous Algebraic CSP} is parameterized by four integers $M$, $k$, $d$, and $m$, along with a finite field $\mathbb{F}$ and a set of constraints $\mathcal{C}$ of size $\vert\mathcal{C}\vert = M$. The constraints and assignments are defined as follows:
    \begin{enumerate}
        \item Each constraint $C \in \mathcal{C}$ is of the form $C = (\mat p^{(1)}, \ldots, \mat p^{(k)}, H)$, where $\mat p^{(1)}, \ldots, \mat p^{(k)} \in \mathbb{F}^m$ are points and $H \subseteq \mathbb{F}^k$ is a linear subspace.
        \item An \emph{assignment polynomial $f$} for the HomAlgCSP is a polynomial $f : \mathbb{F}^m \rightarrow \mathbb{F}$.
        \item We say that a constraint $C \in \mathcal{C}$ is \emph{satisfied} by the assignment polynomial $f$ iff the vector $f(\mat p^{(1)}) \| \ldots \| f(\mat p^{(k)})$ is orthogonal to $H$.
    \end{enumerate}
    The goal is to find a \emph{nonzero} assignment polynomial $f$ of degree at most $d$ satisfying the maximum number of constraints.
\end{problemnormal}

More precisely, we reduce to a gap version of the problem:

\begin{problemnormal}[GapAlgCSP]
    Given a Homogeneous Algebraic CSP with parameters $M, k, d,$ and $m$ over a finite field $\mathbb{F}$, along with additional parameters $1 \geq \alpha > \beta \geq 0$, the $(\mathbb{F}, M, k, d, m, \alpha, \beta)$-GapAlgCSP problem is to distinguish between the following two cases:
    \begin{enumerate}
        \item There exists a nonzero assignment polynomial of degree at most $d$ satisfying at least an $\alpha$ fraction of the constraints.
        \item Every nonzero assignment polynomial of degree at most $1000d$ satisfies at most a $\beta$ fraction of the constraints.
    \end{enumerate}
\end{problemnormal}

Notice that a gap is present in \emph{two} regards. First, we want the assignment polynomial in case (1) to satisfy a larger fraction of constraints than any possible assignment polynomial in case (2). Additionally, we want case (2) to hold even for assignment polynomials of degree a factor of $1000$ larger than in case (1).

Theorem 3.4 and Remark 3.5 from \cite{khot2006ruling} can be stated formally as a reduction from GapMDC to GapAlgCSP, as confirmed by Khot \cite{khotpersonal}:


\begin{theorem}[\cite{khot2006ruling}] \label{thm:khotstart}
    Let $N$ and $M$ be any positive integers, let $\alpha, \beta \in [0,1]$ be any real numbers, let $m^*$ and $d^*$ be any positive integers satisfying $\binom{m^*}{d^*} \geq \max(N, M)$, and let $\lambda$ be an integer satisfying $N \leq 2^\lambda \leq N^2$. There is a deterministic $\exp{O((m^*d^*)^{O(1)})}$ time\footnote{This running time is just a very loose upper bound.} reduction from
    \[(\mathbb{F}_{2^\lambda}, N, M, \alpha, \beta)\text{-GapMDC}\]
    to
    \[(\mathbb{F}_{2^\lambda}, M', k = 21, d', m', \alpha', \beta')\text{-GapAlgCSP,}\]
    where
    \begin{align*}
        M' & = \exp{O((m^*d^*)^{O(1)})} \\
        d' & = O(d^*) \\
        m' & = O((m^*d^*)^{O(1)}) \\
        \alpha' & = 1 - \alpha \\
        \beta' & = \max((1 - \beta), O(d^*/2^\lambda)).
    \end{align*}
\end{theorem}


Now consider starting with a SAT instance of size $n$ and applying our reduction from Theorem~\ref{thm:bigMDC}, which gives an instance of
\[(\mathbb{F}_{2^\lambda}, N, M, 0.5, (1-n^{-\Omega(\log\log n)}))\text{-GapMDC},\]
where $N = n^{\Theta(\log\log n)}$, $M = n^{O(\log\log n)}$, $N \leq 2^\lambda \leq N^2$, and additionally $\lambda$ is a power of $2$. We wish to reduce from this problem to an instance of GapAlgCSP by means of Theorem \ref{thm:khotstart}. Setting
\[m^* = \lceil2^{\sqrt{\log n}\log^4\log n}\rceil\]
and
\[d^* = \lceil\sqrt{\log n}/\log^2\log n\rceil,\]
we have
\[\binom{m^*}{d^*} \geq \left(\frac{m^*}{d^*}\right)^{d^*} \geq 2^{\Omega(\log n \log^2\log n)} \geq n^{\Omega(\log^2\log n)},\]
which is greater than $N$ and $M$, assuming $n$ is sufficiently large. This gives the following corollary:

\begin{corollary} \label{cor:getAlgCSP}
    There is a deterministic $\exp{2^{O(\sqrt{\log n}\log^4\log n)}}$ time reduction from SAT instances of size $n$ to
    \[(\mathbb{F}_{2^\lambda}, M, k = 21, d, m, \alpha = 0.5, \beta = n^{-\Omega(\log\log n)})\text{-GapAlgCSP},\]
    where
    \begin{align*}
        \lambda & = \Theta(\log n\log\log n) \text{ is a power of 2} \\
        M & = \exp{2^{O(\sqrt{\log n}\log^4\log n)}} \\
        d & = O(\sqrt{\log n}/\log^2\log n) \\
        m & = 2^{O(\sqrt{\log n}\log^4\log n)}.
    \end{align*}
\end{corollary}

Now, we are nearly done. As confirmed by Khot \cite{khotpersonal}, the reduction from GapAlgCSP to QRDH appearing in \cite{khot2016hardness} yields the following theorem:

\begin{theorem}[\cite{khot2016hardness}] \label{thm:tohypergraph}
    Let $M, d,$ and $m$ be any positive integers, let $\alpha, \beta \in [0, 1]$ be any real numbers, and let $\lambda$ and $\lambda'$ be any positive integers such that $\lambda'$ divides $\lambda$. There is a deterministic $(M2^\lambda)^{O((md)^{O(1)})}$ time\footnote{This running time is just a very loose upper bound.} reduction from
    \[(\mathbb{F}_{2^\lambda}, M, k = 21, d, m, \alpha, \beta)\text{-GapAlgCSP}\]
    to
    \[\left((M2^\lambda)^{O((md)^{O(1)})}, (M2^\lambda)^{O((md)^{O(1)})}, O(d), 2^{\lambda'}, \alpha, \beta^{\Omega(1)} + 2^{-\Omega(\lambda)}\right)\text{-QRDH}.\]
\end{theorem}

The final step is to combine Corollary \ref{cor:getAlgCSP} with Theorem \ref{thm:tohypergraph}. To this end, let $\lambda'$ be a power of two in the range $[\sqrt{\log n}, 2\sqrt{\log n})$, and set $r = 2^{\lambda'}$. We know that, assuming $n$ is sufficiently large, $\lambda'$ divides any possible value of $\lambda$ from Corollary \ref{cor:getAlgCSP} (since $\lambda$ is also a power of two), so we recover Theorem \ref{thm:khotqrdh}:

\paragraph{Theorem \ref{thm:khotqrdh}, Restated.}
    \emph{There is a deterministic $\exp{2^{O(\sqrt{\log n}\log^4\log n)}}$ time reduction from SAT instances of size $n$ to 
    \begin{align*}
        \big( & N = \exp{2^{O(\sqrt{\log n}\log^4\log n)}}, \\
        & M = \exp{2^{O(\sqrt{\log n}\log^4\log n)}}, \\
        & d = O(\sqrt{\log n}/\log^2\log n), \\
        & r = 2^{\Theta(\sqrt{\log n})}, \alpha = 0.5, \beta = n^{-\Omega(\log\log n)}\big)\textnormal{-QRDH}.
    \end{align*}}

\section*{Acknowledgments}
We gratefully acknowledge Subhash Khot for his inputs, and especially for confirming the correctness of our applications of his theorems. This research was supported in part from a Simons Investigator Award, DARPA expMath award, NSF grant 2333935, BSF grant 2022370, a Xerox Faculty Research Award, a Google Faculty Research Award, an Okawa Foundation Research Grant, and the Symantec Chair of Computer Science. This material is based upon work supported by the Defense Advanced Research Projects Agency through Award HR001126CE054.

\newpage
\bibliographystyle{alpha}
\bibliography{refs.bib}
\newpage

\appendix

\section{Proof of Claim \ref{claim:calculategamma}} \label{app:proof:calculategamma}

Recall the statement of the claim:

\paragraph{Claim \ref{claim:calculategamma}, Restated.}
    \emph{Let $p \geq 1$ be any constant, let $M$ be a parameter such that $1 \leq M \leq$ \\ $\exp{2^{O(\sqrt{\log n}(\log\log n)^{O(1)})}}$, and let $M' = \exp{n^{1/\log\log n}\log M}$. Then
        \[2^{n^{1/\log\log n}/p - 1} \geq \exp{\Omega\left((\log M')^{1 - (\log\log\log M')^{O(1)}/\sqrt{\log\log M'}}\right)}.\]}

\begin{proof}
    Taking logs, we have the following, assuming that $n$ is sufficiently large and we take $c$ to be an appropriate constant.
    \[\log n/\log\log n \quad \leq \quad \log\log M' \quad \leq \quad \sqrt{\log n}(\log\log n)^c + \log n/\log\log n\]
    Again assuming $n$ is sufficiently large, further manipulation gives:
    \begin{align*}
        \sqrt{\log\log M'}(\log\log\log M')^{c+2} \quad \geq & \quad \sqrt{\log n/\log\log n} \cdot (\log\log n - \log\log\log n)^{c+2} \\
        \geq & \quad \sqrt{\log n}(\log\log n)^c
    \end{align*}
    Therefore we can write:
    \begin{align*}
        \log\log M' - \sqrt{\log\log M'}(\log\log\log M')^{c+2} \quad \leq & \quad \log n/\log\log n \\
        \log\log M'\left(1 - (\log\log\log M')^{c+2}/\sqrt{\log\log M'}\right) \quad \leq & \quad \log n/\log\log n \\
        \exp{(\log M')^{1 - (\log\log\log M')^{c+2}/\sqrt{\log\log M'}}} \quad \leq & \quad \exp{n^{1/\log\log n}} \\
        \exp{\Omega\left((\log M')^{1 - (\log\log\log M')^{O(1)}/\sqrt{\log\log M'}}\right)} \quad \leq & \quad \exp{n^{1/\log\log n}/p - 1} \tag{Using that $p \geq 1$ is a constant.}
    \end{align*}
\end{proof}

\section{Proof of Proposition \ref{claim:avgprinciple}} \label{app:claim:implicateR}

Below, we re-state the proposition.

\paragraph{Proposition \ref{claim:avgprinciple}, Restated.}
        \emph{Let $n$, $M$, and $N$ be positive integers, and let
        \begin{align*}
            \alpha & = 0.5 \\
            r & = 2^{\Theta(\sqrt{\log n})} \\
            \log n/\log\log n \leq q & \leq 2\log n/\log\log n \\
            1 \leq d & \leq O(\sqrt{\log n}/\log^2\log n) \\
            h & \leq (2\alpha(1/r)^{d-1}M)^q \\
            \left(\frac{\alpha(1/r)^{d-1}M}{N/r}\right)^q/n \leq w &
        \end{align*}
        Then assuming $n$ is sufficiently large,
        \[\frac{2h \cdot d^q}{N^q\left(\frac{(\alpha(1/r)^{d-1})^q}{n^2}\right)^{1/d}} < w.\]}

    \begin{proof}
    \begin{align*}
        \frac{2h \cdot d^q}{N^q\left(\frac{(\alpha(1/r)^{d-1})^q}{n^2}\right)^{1/d}} \leq & \frac{2h \cdot (O(\sqrt{\log n}/\log^2\log n))^{2\log n/\log\log n}}{N^q\left(\frac{(\alpha(1/r)^{d-1})^q}{n^2}\right)^{1/d}} \tag{$d = O(\sqrt{\log n}/\log^2\log n)), q \leq 2\log n/\log\log n$} \\
        \leq & \frac{hn^{O(1)}}{N^q\left(\frac{(\alpha(1/r)^{d-1})^q}{n^2}\right)^{1/d}} \\
        \leq & \frac{hn^{O(1)}}{N^q(\alpha(1/r)^{d-1})^{q/d}/n^{O(1)}} \tag{$d \geq 1$} \\
        \leq & \frac{hn^{O(1)}}{N^q\alpha^{q/d}(1/r)^{q-q/d}} \\
        \leq & \frac{hn^{O(1)}}{N^q(1/2)^{2\log n/\log\log n}(1/r)^{q-q/d}} \tag{$\alpha = 0.5, q \leq 2\log n/\log\log n, d \geq 1$} \\
        \leq & \frac{hn^{O(1)}}{N^q(1/r)^{q-q/d}} \\
        \leq & \frac{(2\alpha(1/r)^{d-1}M)^qn^{O(1)}}{N^q(1/r)^{q - q/d}} \tag{$h \leq (2\alpha(1/r)^{d-1}M)^q$} \\
        \leq & \frac{2^{2\log n/\log\log n}(\alpha(1/r)^{d-1}M)^qn^{O(1)}}{N^q(1/r)^{q - q/d}} \tag{$q \leq 2\log n/\log\log n$}\\
        \leq & \left(\frac{\alpha(1/r)^{d-1}M}{N/r}\right)^q \cdot \frac{n^{O(1)}}{(1/r)^{-q/d}} \\
        \leq & wn \cdot \frac{n^{O(1)}}{(1/r)^{-q/d}} \tag{$w \geq \left(\frac{\alpha(1/r)^{d-1}M}{N/r}\right)^q/n$}\\
        \leq & w \cdot \frac{n^{O(1)}}{r^{q/d}} \\
        \leq & w \cdot \frac{n^{O(1)}}{\left(2^{\Theta(\sqrt{\log n})}\right)^{(\log n/\log\log n)/(O(\sqrt{\log n}/\log^2\log n))}} \tag{$r = 2^{\Theta(\sqrt{\log n})}, q \geq \log n/\log\log n, d = O(\sqrt{\log n}/\log^2\log n$}\\
        \leq & w/2^{\Omega(\log n\log\log n)}\\
        < & w \tag{Assuming $n$ sufficiently large.}
    \end{align*}
    \end{proof}

\end{document}